%% file: main.tex
\documentclass[11pt]{article}
\usepackage{fullpage}
\usepackage{authblk}
\usepackage{float}
\usepackage[hidelinks]{hyperref}
\usepackage{amsmath}
\usepackage{amssymb}
\usepackage{amsthm}
\usepackage{thm-restate}
\usepackage{enumitem}
\usepackage[dvipsnames]{xcolor}
\usepackage[numbers]{natbib}
\usepackage{appendix}
\usepackage{mathtools}
\allowdisplaybreaks
\usepackage[ruled,vlined]{algorithm2e}

\input{notation.tex}

\newtheorem{theorem}{Theorem}
\newtheorem{lemma}[theorem]{Lemma}
\newtheorem{corollary}[theorem]{Corollary}
\newtheorem{definition}[theorem]{Definition}
\newtheorem*{definition*}{Definition}

\title{Polynomial-Time Approximation Schemes via Utility Alignment: Unit-Demand Pricing and More\footnote{This work was supported by NSF awards CCF 2218813 and CCF 2218814.}}
\date{}
\author[1]{Robin Bowers\thanks{\texttt{robin.bowers@colorado.edu}}}
\author[2]{Marius Garbea\thanks{\texttt{mgarbea@drexel.edu}}}
\author[2]{Emmanouil Pountourakis\thanks{\texttt{manolis@drexel.edu}}}
\author[3]{Samuel Taggart\thanks{\texttt{staggart@oberlin.edu}}}
\affil[1]{University of Colorado, Boulder}
\affil[2]{Drexel University}
\affil[3]{Oberlin College}
\begin{document}

\maketitle

\input{abstract}

\section{Introduction}
\label{sec:introduction}
\input{introduction/introduction}

\section{Related Work}
\label{sec:related-work}
\input{related/related_work}

\section{PTAS for Utility Configuration}
\label{sec:main-results}
\input{main_results/main_results}

\section{Applications and Utility Alignment Analyses}
\label{sec:applications}
\input{applications/applications}

\section{Unit-demand Pricing}
\label{sec:pricing}
\input{pricing/pricing}

\newpage
\bibliographystyle{plainnat}
\bibliography{references}

\newpage
\appendix
\section*{Appendix}
\input{appendix/appendix}\label{sec:appendix}

\end{document}

%% file: notation.tex
\newcommand{\set}[1]{\left\{#1\right\}}

\newcommand{\E}{\mathbb{E}}
\newcommand{\I}{\mathbb{I}}
\newcommand{\R}{\mathbb{R}}
\newcommand{\prob}{\text{Pr}}
\newcommand{\opt}{\text{OPT}}

\newcommand{\OPT}{\mathrm{OPT}}
\newcommand{\utilAgent}[1]{u_{#1}^A}
\newcommand{\utilPrincipal}[1]{u_{#1}^P}

\newcommand{\optionCount}{n}
\newcommand{\agentPreferredOption}{i^*}
\newcommand{\bucket}{B}
\newcommand{\bucketsCount}{M}

\newcommand{\principalExpectation}[1]{\E[\utilPrincipal{#1}]}
\newcommand{\expectedUtilityBelow}{X}
\newcommand{\utilityThreshold}{U}
\newcommand{\bucketProbabilities}{q}
\newcommand{\bucketConditionals}{q^{\leq}}
\newcommand{\bucketConditionalsExpr}[2]{\Pr[\utilAgent{#1} \in \bucket_{#2}~|~\utilAgent{#1} \in \bucket_{\leq #2}]}
\newcommand{\bucketProbabilitiesLB}{c}
\newcommand{\bucketExpectationExpr}[2]{\E[\utilPrincipal{#1}~|~\utilAgent{#1} \in \bucket_{#2}]}
\newcommand{\q}{q}
\newcommand{\config}{\mathbf C}
\newcommand{\configi}{C_i}
\newcommand{\objective}{\overline{\text{Util}}}
\newcommand{\coef}{\bucketProbabilitiesLB}

\newcommand{\optua}{u^{A}_{\text{OPT}}}
\newcommand{\optup}{u^{P}_{\text{OPT}}}

\newcommand{\bucketprobalg}{\bucketConditionalsExpr{\text{ALG}}{j}}
\newcommand{\bucketprobopt}{\bucketConditionalsExpr{\text{OPT}}{j}}
\newcommand{\bucketprob}{\bucketConditionalsExpr{}{j}}
\newcommand{\bucketprobi}{\bucketConditionalsExpr{i}{j}}

\newcommand{\qleqj}{\bucketProbabilities_{j}^{\leq}}
\newcommand{\qjest}{\bucketProbabilities_{j}^{\text{est}}}
\newcommand{\qjbarest}{\overline{\bucketProbabilities}_{j}^{\text{est}}}
\newcommand{\qleqij}{\bucketProbabilities_{ij}^{\leq}}
\newcommand{\qleqbarj}{\overline{\bucketProbabilities}_{j}^{\leq}}
\newcommand{\qleqbarij}{\overline{\bucketProbabilities}_{ij}^{\leq}}

\newcommand{\uleqj}{u_j^\leq}
\newcommand{\uestj}{u_j^{\text{est}}}

\newcommand{\val}{v}
\newcommand{\bias}{b}

\newcommand{\rev}{\mathcal{R}}
\newcommand{\price}{p}
\newcommand{\prices}{\mathbf \price}
\newcommand{\priceSet}{\mathcal{P}}
\newcommand{\umin}{u_{\min}}
\newcommand{\umax}{u_{\max}}
\newcommand{\optrev}{p_{i^*({\prices})}}


%% file: abstract.tex
\begin{abstract}
    This paper derives polynomial-time approximation schemes for several NP-hard stochastic optimization problems from the algorithmic mechanism design and operations research literatures.
    The problems we consider involve a principal or seller optimizing with respect to a subsequent choice by an agent or buyer.
    These include posted pricing for a unit-demand buyer with independent values (\citet{CHK07,CD11}), assortment optimization with independent utilities (\citet{TR04}), and delegated choice (\citet{KPT24}).
    Our results advance the state of the art for each of these problems.
    For unit-demand pricing with discrete distributions, our multiplicative PTAS improves on the additive PTAS of \citet{CD11}, and we additionally give a PTAS for the unbounded regular case, improving on the latter paper's QPTAS.
    For assortment optimization, no constant approximation was previously known.
    For delegated choice, we improve on both the $3$-approximation for the case with no outside option and the super-constant-approximation with an outside option.
    
    A key technical insight driving our results is an economically meaningful property we term {\em utility alignment}.
    Informally, a problem is utility aligned if, at optimality, the principal derives most of their utility from realizations where the agent's utility is also high.
    Utility alignment allows the algorithm designer to focus on maximizing performance on realizations with high agent utility, which is often an algorithmically simpler task.
    We prove utility alignment results for all the problems mentioned above, including strong results for unit-demand pricing and delegation, as well as a weaker but very broad guarantee that holds for many other problems under very mild conditions.
\end{abstract}

%% file: introduction/introduction.tex
We consider algorithms for a family of agent-principal problems from economics and operations research.
The problems we consider involve a principal (e.g.\ a seller) presenting a choice to an agent (e.g.\ a buyer).
The principal offers a set of options to the agent, who selects their favorite.
The two main hallmarks of the problems we consider are: (1) the agent's preferences differ from those of the principal, and (2) the agent's preferences are unknown (and modeled probabilistically).
The principal's goal is to maximize their expected utility with respect to the agent's choice, in expectation over their uncertain preferences.

Several well-studied problems satisfy the hallmarks above.
For example, in pricing problems with a single buyer, the seller chooses a price $p_i$ for each item $i$.
The buyer's choice is dictated by their values for the different items, which are unknown to the seller.
The seller seeks to maximize revenue, taken in expectation over the buyer's utility-maximizing choice.
Another example is the model of delegated choice of \citet{KPT24}, inspired by similar classical problems in economics.
In the delegated choice problem, a principal presents a menu of options to an agent, who selects their favorite.
This captures a broad range of deterministic mechanisms for single-agent problems.
Finally, in assortment optimization, a fundamental problem from operations research, a seller chooses an inventory of fixed-price items to stock, again for a single buyer.

We consider versions of these problems where the agent chooses a single option, and where their utilities for the different options are independently distributed.
For pricing, this corresponds to the canonical unit-demand pricing problem studied in \citet{CHK07} and \citet{CD11}.
For assortment optimization, this single-choice random-utility model generalizes the well-studied multinomial logit model of e.g.\ \citet{TR04}.

Prior to this work, none of the problems discussed above had polynomial-time approximation schemes.
The strongest prior approximation results are for unit-demand pricing.
For unconditional multiplicative approximation, \citet{CHK07} give a constant-approximation, which was subsequently improved by \citet{JL24}.
Meanwhile, \citet{CD11} give an additive PTAS for discrete distributions and a multiplicative QPTAS for unbounded distributions satisfying the standard Myerson regularity condition, which improves to a multiplicative PTAS only when the instance satisfies the much stronger condition of monotone hazard rate distributions.
We improve these results to a multiplicative PTAS which holds for both discrete and unbounded regular distributions.
For delegation and assortment optimization, only a handful of even constant approximations were known, and for some variants of the problems, not even that.

Two main hurdles complicate the design of a PTAS for the problems we consider.
First, the typical approach to designing a PTAS is to coarsen the space of values involved, usually by rounding to a grid.
For problems where the algorithm's performance is determined by the choice of an agent, this approach is especially problematic.
While the designer would like to coarsen the space of agent utilities, even small changes to the agent utility can cause them to switch their choice (e.g.\ from a high-priced, high-value item to a low-priced but low-value item), leading to large changes in principal utility.
For pricing, \citet{CD11} resolved this issue with careful discretization schemes.
The second hurdle is that even with careful discretization, existing approximation schemes all required enumerating over an exponentially-sized set of objects, such as utility distributions associated with solutions.
Indeed, the runtime exponent of the unit-demand pricing algorithm in \citet{CD11} depends on the number of discretized prices, which is in turn logarithmic in the number of items, leading to a QPTAS.

To solve these issues, our main technical innovation is to show that for all of the optimization problems we study, the principal can ignore outcomes where the agent's utility is low.
This takes the form of a property we call {\em utility alignment}.
In more detail, utility alignment states that in the optimal solution, most of the principal's utility comes from realizations that have a strong quantile in the agent's utility distribution.
Since the principal and agent tend to get high utility from the same set of realizations, their utilities are aligned.
Utility alignment implies that near-ties are not important for the principal utility: if the agent's utility for different options is independent, then the probability that multiple of these utilities fall in the same high-quantile subinterval is very small: e.g.\ the probability that action $i$ and $j$ both fall into the top $1/n$ quantile of the agent's utility distribution is at most $1/n^2$.
Hence, near-ties mostly occur at low quantiles in the agent's utility distribution, and by utility alignment, these quantiles are unimportant to the principal.
Algorithmically, this means that many of the probability computations in the algorithm can be coarsend to first-order approximations that ignore ties, and can be optimized without enumerating an exponentially-sized set.
The rest of the introduction makes these ideas and problems above more precise.

\subsection{Utility Configuration}\label{sec:intro-utility-configuration}
\input{introduction/utility_configuration}
\subsection{Utility Alignment and Main Results}\label{sec:intro-main-results}
\input{introduction/main_results}

%% file: introduction/utility_configuration.tex
The problems mentioned above are all special cases of a general template problem, which we call {\em utility configuration}.
In utility configuration, a principal (e.g.\ the seller) considers a set of $n$ possible actions that the agent will subsequently choose from.
Each option can be ``configured'' one of several ways, corresponding to setting the price of an item, or including/excluding an action from a delegation set.
This configuration then determines the distribution over agent and principal utilities for this item.
Given the configurations, the agent realizes their utilities (and hence the principal utilities as well) and picks their favorite action.
We define the problem below, before defining pricing, delegation, and assortment optimization, and showing that these are special cases.

\begin{definition}[Utility Configuration]\label{def:uc}
    An instance of the {\em utility configuration problem} is given by $n$ {\em actions}.
    Each action $i$ has at most $m$ {\em configurations}.
    Configuration $j$ for action $i$ induces a joint distribution $F_{ij}$ over pairs $(u_i^A,u_i^P)$, of agent and principal utilities.
    A solution to the utility configuration problem (a {\em configuration}) is a choice $\config=(C_1,\ldots,C_n)$ of configuration $C_i\in\{1,\ldots,m\}$ for each action $i$.
    Given a configuration $\config$, the agent draws $(u_i^A,u_i^P)$ independently from $F_{iC_i}$ for each action $i$.
    Denote their favorite action given the realized utilities by the random variable $i^*(\config)= \arg\max_{i} u_i^A$.
    The principal's objective is to maximize their expected utility, given by $\E[u_{i^*(\config)}^P]$.
\end{definition}

\paragraph{Remark on notation.} When context makes the choice of configuration clear, we often omit the subscript $i^*(\config)$ to denote the agent's preferred choice, and instead denote the utilities from the agent's favorite choice by $u^A=u_{i^*(\config)}^A$ and $u^P=u_{i^*(\config)}^P$.\\

For computational purposes, we consider {\em discrete} instances of utility configuration, where the distributions $F_{ij}$ are discrete distributions with finite supports for every $i$ and $j$, given as input in the form of a list of mass points with respective probabilities.
In principle, though, it is also natural to consider variants with continuous utility distributions, as we do for unit-demand pricing at certain points in the paper.
To enable approximation, we restrict to nonnegative principal utilities $u_i^P\geq 0$, but do not need to place any such restriction on agent utilities, as only the order of agent utilities matters to the principal. 

We now formalize the three problems discussed in the introduction --- pricing, delegation, and assortment optimization --- and show how they can be cast as utility configuration problems.

\paragraph{Delegated Choice.} In the delegated choice problem of \citet{KPT24}, a principal must select a subset $S\subseteq\{1,\ldots,n\}$ of actions to present to an agent.
The agent realizes their utilities for the actions, and chooses their favorite.
\citet{KPT24} consider a simple utility model, where the principal's utility is given by $v_i$, which is unknown to the principal and distributed according to some $F_i$, and the agent's utility is $v_i+b_i$, for a bias term which is constant for each $i$ and known to the principal.
This can be cast as utility configuration in the following way: each action $i$ can be configured as {\em in}, in which case $(u_i^A,u_i^P)=(v_i+b_i,v_i)$ with $v_i\sim F_i$, or {\em out}, in which case $(u_i^A,u_i^P)=(-\infty,0)$.
One can define further variants of this problem, e.g.\ where the biases $b_i$ are also random, or where the agent has an outside option, an action $0$ that cannot be excluded by the principal.
We consider these additional variants in Section~\ref{sec:applications}, and reduce them to utility configuration as well.

\paragraph{Unit-demand Pricing.} In the unit-demand pricing problem, a seller prices $n$ items for consideration by a unit-demand buyer.
The buyer's value $v_i$ for each item $i$ is drawn independently from a known distribution $F_i$.
Given a pricing $\prices = (p_1,\ldots,p_n)$, the buyer picks the item that maximizes their utility $v_i-p_i$ (yielding revenue $p_i$ for the seller).
If no item gives positive utility, they choose not to purchase, yielding revenue $0$ for the seller.

This can be framed as a utility configuration problem in the following way: each item $i$ corresponds to an action, and each possible price $p_i$ for item $i$ corresponds to a configuration for that item.
Notice that this implies an infinite number of configurations per item.
While this is well-defined according to Definition~\ref{def:uc}, it clearly poses computational issues, which we show how to circumvent in Section \ref{sec:pricing}.
In the configuration for action $i$ corresponding to price $p_i$, the utilities are given by $(u_i^A,u_i^P)=(v_i-p_i,p_i)$ if $v_i-p_i\geq 0$, or $(u_i^A,u_i^P)=(0,0)$ if $v_i-p_i<0$, with $v_i\sim F_i$.

\paragraph{Assortment Optimization.} In assortment optimization, a seller chooses a subset $S\subseteq\{1,\ldots,n\}$ of fixed-price items to stock for a unit-demand buyer, who wishes to purchase at most one item.
Item $i$ has price $p_i$, and yields revenue $p_i$ to the seller if $i$ is purchased.
While there is an extensive literature solving this problem for a dizzying range of choice models for the buyer, we consider a simple {\em utility-based model} with independent utilities, which is natural to state and generalizes the best-known multinomial logit model (\citet{TR04}) as well as variants (\citet{AFS23}).
Using notation similar to unit-demand pricing, we have the following: for each item $i$, the buyer draws a value $v_i$ independently from a distribution $F_i$.
They select the item maximizing $v_i-p_i$ from the set $S$ of available items.
One can consider a variant of the problem where the buyer chooses not to purchase if $v_i-p_i<0$ for all $i\in S$.
In this case, for $i\in S$, we have $(u_i^A,u_i^P)=(v_i-p_i,p_i)$ if $v_i-p_i\geq 0$, or $(u_i^A,u_i^P)=(0,0)$ if $v_i-p_i<0$, and for $i\notin S$, $(u_i^A,u_i^P)=(-\infty,0)$.
It is more common in the literature, however, to consider the no-purchase option as yielding random utility $u_0$ for the buyer.
We discuss how to incorporate this variant in Section~\ref{sec:applications-assortment}.

%% file: introduction/main_results.tex
We give polynomial-time approximation schemes for all problems mentioned in the previous section.
In fact, we give a single algorithm, an approximation scheme for utility configuration, which can be applied via the reductions described above.
The algorithm's runtime is polynomial in the size of the utility configuration instance, and its performance is parametrized by a fineness parameter $M$ (which also governs the polynomials in the runtime) and the {\em utility alignment} of the instance, which is a parameter we define formally below.
Intuitively, a utility configuration instance is utility-aligned if, in the optimal solution, the agent and principal derive most of their utility from the same realizations.
More precisely, when we consider the bottom $q$-quantile of the agent's utility distribution in the optimal configuration, it has to be that this quantile range does not constitute too much of the principal's overall expected utility.
Equivalently, the principal's expected utility conditioned on this range of agent utilities should not be too much higher than their expected utility overall.
Formally, we define utility alignment as follows:
\begin{definition}[Utility Alignment]\label{def:alignment}
    Given a non-increasing function $f$, an instance of utility configuration is {\em$f(q)$-utility aligned} if for any utility threshold $U$ such that in the optimal configuration $\config^*$, $\prob[u^A\leq U]=q$, the principal utilities under $\config^*$ satisfy
    \begin{equation*}
        \E[u^P~|~u^A\leq U]\leq f(q) \E[u^P],
    \end{equation*}
    where $u^A$ and $u^P$ are the agent and principal utilities from the agent's preferred action under $\config^*$.
\end{definition}

For intuition, first consider delegation and unit-demand pricing.
In Sections \ref{sec:applications-delegation} and \ref{sec:applications-pricing}, we prove the following:
\begin{lemma}
    The delegation problem (with deterministic biases and no outside option) and unit-demand pricing are both $2$-utility aligned.
\end{lemma}
Consequently, under the optimal solution to each of these problems, the principal's utility conditioned on the bottom $q$-quantile of agent utilities is no more than twice their overall expected utility.
Equivalently, the contribution to their overall utility from the bottom $q$ quantile range is at most a $2q$-fraction overall.
For example, the bottom $10\%$ of agent utilities can contribute at most $20\%$ of the principal's optimal utility.
We emphasize that utility alignment is a property of optimal solutions: it is easy to find non-optimal solutions with extremely misaligned utilities, e.g.\ where nearly all of the principal's utility comes from the agent's least favorite realizations.

For a more complicated example, we prove the following theorem in Section~\ref{sec:applications-general}, which we state very informally here:
\begin{lemma}[Informal]
    If agent and principal utilities are positively correlated for each individual action, then the instance satisfies $4/\sqrt q$-utility alignment.
\end{lemma}
This lemma will hold for the more complicated variants of delegated choice discussed above, as well as for assortment optimization with independent utilities.
Notice that $4/\sqrt{q}$-utility alignment is weaker: it implies that, conditioned on the very weakest quantiles of the agent's utility, the principal's utility may be very large.
However, it still implies that not too much of the principal's utility comes from low agent quantiles.
For example, for $q=.01$, no more than $q\cdot 4/\sqrt q\cdot \E[u^P]=.01\cdot4/.1\cdot \E[u^P]=.4\cdot \E[u^P]$ can come from the bottom $q$-quantile, i.e.\ the bottom $1\%$ of agent utilities can generate no more than $40\%$ of the principal's utility.

The performance of our approximation scheme is parametrized by the level of utility alignment of the utility configuration instance.
At a high level, the algorithm coarsens the quantile space of agent utilities into $M$ bins, then uses additive approximations to compute probabilities and ensure that the agent's utility falls into each bin with probability approximately $1/M$.
It also maintains additive estimates of the principal's expected utility conditioned on each bin.
These additive estimates will be inaccurate for low quantiles of agent utility, but accurate for high ones.
Utility alignment then implies that these larger errors from low quantiles do not seriously harm the algorithm's performance.
We obtain the following result, stated informally.
\begin{theorem}[Informal]
    Let $M$ be a positive integer.
    There is an algorithm that runs in polynomial time for any utility configuration instance, such that if the instance is $2$-utility aligned, the algorithm is a $(1-O(\log(M)/M))$-approximation, and if the instance is $4/\sqrt{q}$-utility aligned, the algorithm is $(1-O(1/\sqrt{M}))$-approximation.
\end{theorem}
More generally, our approximation scheme is a PTAS for any instances that are $c/q^{1-\epsilon}$-utility aligned for some $c,\epsilon>0$.

\subsection{Structure of the Paper}

We discuss previous work on the various problems we consider in Section~\ref{sec:related-work}.
Section~\ref{sec:main-results} then gives our algorithm for utility configuration, and analyzes its performance as a function of the number of bins (the ``fineness parameter'') $M$ and the level of utility alignment in the instance.
In Section~\ref{sec:applications}, we then consider the applications of our framework.
In addition to the problems defined in Section~\ref{sec:intro-utility-configuration}, we show how to reduce several other variants to the utility configuration problem.
For each, we prove a bound on the utility alignment.
For all problems except unit-demand pricing, these immediately imply that the algorithm of Section~\ref{sec:main-results} is a PTAS.
The one exception is unit-demand pricing.
Because the set of prices is a continuum, a straightforward reduction would yield an infinite number of configurations per action.
Instead, we apply a price discretization result from \citet{CD11}, which implies that polynomially many prices suffice to obtain near-optimal solutions.
This then combines with our results to produce a multiplicative PTAS for the discrete version of unit-demand pricing.
Along the way, we also show how to use the value discretization and truncation results from that same paper to obtain a PTAS for the case of unbounded regular distributions as well.
These results can be found in Section~\ref{sec:pricing}.

%% file: related/related_work.tex
Our PTAS applies to a variety of problems, each with its respective body of research.
We group our discussion of related work accordingly.

\paragraph{Unit-demand Pricing.} Both pricing and more general mechanism design for unit-demand buyers have been important testbeds for algorithmic mechanism design.
The most directly related works design approximation algorithms for a unit-demand buyer with independently distributed values.
\citet{CHK07} give a $4$-approximation for this problem, which has recently been improved to a $3$-approximation in \citet{JL24}.
Subsequent work by \citet{CD11} derives an additive PTAS for the discrete version of the problem, a multiplicative PTAS for the version with continuous distributions that may be unbounded but satisfy the MHR assumption, and a multiplicative QPTAS for the unbounded regular case.
Additionally, they prove several powerful truncation and extreme value results that we use in our analysis of unit-demand pricing in Section~\ref{sec:pricing}.
Our results are the first multiplicative PTAS for both the discrete and unbounded regular cases.
On the hardness side, \citet{CDPSY14} prove that even with independent values, the discrete version of the problem is NP-hard.
With correlation across values, \citet{BK07} further prove logarithmic hardness-of-approximation results; our techniques have no bearing on this much harder setting.

Several natural extensions beyond pricing with independent values have been well-studied.
In revenue management and operations research, past work such as \citet{AFMZ04,RRG06} has considered more general demand models.
There has also been extensive work on the more general mechanism design/lottery pricing problems, where the principal can populate their menu with distributions over items (also called lotteries) for the buyer to choose from.
This work includes \citet{KSMSW19}, who derive near-optimal mechanisms, and \citet{CDOPSY15}, who prove computational hardness.
An interesting question is whether our techniques have bearings on any of these settings.
Other work has considered multi-buyer versions of the problem, where the seller designs sequential posted pricing mechanisms or more general multi-buyer mechanisms. In these settings, welfare is also a non-trivial objective.
These include \citet{CHMS10, Yao15, CM16, CZ17, DKL20, DFKL20}.
Finally, one can consider other families of combinatorial valuations for items, such as additive or subadditive valuations.
Some papers on such models include \citet{HN13, LY13, BILW20, RW18, CM16, CZ17, DKL20}.

\paragraph{Delegation.}
Delegation is a classical problem in microeconomic theory (\citet{Holmstrom77}).
Typical work from economics focuses on deriving explicit optimal delegation mechanisms in settings where the set of actions is the real line, and where the agent and principal have single-peaked preferences with differing biases between the agent and principal.
Papers in this vein include \citet{MS06}, \citet{MS91}, \citet{AM08}, and \citet{AB12}.
The combinatorial model we consider is due to \citet{KPT24}, who formulate a richer model more amenable to algorithmic analysis.
They prove that simple threshold policies are a $3$-approximation under independent utilities and no outside option, and give a parametrized super-constant approximation when the agent has an outside option that the principal can't exclude.
They also show that even with no outside option, the case with independent utilities is NP-hard, and the case with correlated utilities is APX-hard. 
We obtain a PTAS for both settings, as well as the more general setting where biases are also random.

Recent work has focused on connections between the delegated project choice model of \citet{AV10} and classic optimal search results like prophet inequalities (\citet{KS77,S84}) and Pandora's box (\citet{W79}).
These include \citet{KK18} and followups \citet{BD21, BD24, BDP22}.
These papers observe that for delegation problems where the principal can verify their utility from the agent's choice before accepting the agent's selected option, the principal can approximate the utility they would obtain if they controlled the agent's behavior directly.
Other papers have included an online variant of the problem (\citet{BHHS24}) and a multi-agent variant (\citet{SRH23}).
In the \citet{KPT24} variant of delegation, the comparisons to first-best made by the papers above are not possible, as the principal is not able to verify their utility for the agent's choice.
Nonetheless, several of our technical lemmas, including the utility alignment result for delegation, at a high level resemble arguments from the prophet inequalities literature.

\paragraph{Assortment Optimization.} Assortment optimization is a fundamental problem in revenue management, and has been studied under a wide range of demand models for the buyer.
The simplest model is the multinomial logit (MNL) model, which can be expressed as a random utility model with independent utilities.
\citet{TR04} show that for this restricted class, simple threshold-based policies (``revenue-ordered assortments'') are optimal.
Other utility models maintain this focus on narrowly-tailored parametrized families: \citet{AFS23} consider the ``exponomial choice model'' (but do not give approximation algorithms), which, like MNL, has independent utilities across items.
Many other papers introduce correlation: \citet{RSTT14} consider a mixture over MNL models (corresponding to correlation across utilities for different items), and prove NP-hardness, while \citet{DGT14,LRT15} consider nested variants of the logit model, and consider both algorithms and hardness.
No prior work has shown that any model with independent utilities is NP-hard.
However, in Appendix~\ref{sec:appendix}, we show that this in fact follows from a similar argument to one in \citet{KPT24}:
\begin{theorem}
    The assortment optimization problem with independent item utilities is NP-hard, even when the no-purchase option gives the buyer utility zero deterministically.
\end{theorem}
No prior work has obtained constant-approximations or better for the case with general independent utilities, either.
Our PTAS is the first such result.

Several other papers have gone beyond heavily parametrized utility models.
\citet{AFLS18} consider choices generated distributions over rankings (equivalently, fully-correlated utilities), and show that the problem is as hard to approximate as the maximum independent set problem.
\citet{RS24} consider the special case of Mallows models for rankings, and give a QPTAS for that problem.
\citet{FT17} consider a variant where preferences are drawn by following a Markov chain.
\citet{BJ20} give approximation algorithms for general choice models, with parametrized performance depending on the model.
Several other papers have considered versions of assortment optimization with additional constraints or multi-item purchases.
These include \citet{BHL24, FPT19, DGZ22, CMSX24, ILMST21}.
The latter two papers are notable for proving hardness results as well.
An interesting direction for future work would be extending our results to some of these related models.

\paragraph{Stochastic Probing.}
At a high level, the problems we consider resemble some of those recently considered in the literature on {\em stochastic probing}.
These papers typically consider an algorithm designer who must choose a limited number of random variables to observe from a larger set.
Upon realizing the values of the chosen random variables, the designer picks the one with the best reward.
\citet{SS21} and \citet{MNPR20} consider a cardinality-constrained version of this problem (the designer must pick $k$ variables), and \citet{GGM06} consider a version with probing costs.
\citet{CHLLLL16} consider a version embedded in a bandits problem.
These models resemble delegation, except that the decision-maker selecting the variables is different from the one observing the realizations and picking the final reward.
Our techniques may have future applications in this literature as well.
Two related papers that bear final mention are \citet{MMW23} and \citet{ATMPZ23}, which consider a regret-minimizing decision-maker who observes only noisy observations of rewards.
The decision-maker knows the noise distributions but not the rewards: there are loose connections to the model of delegation with random bias we introduce in Section~\ref{sec:applications-delegation-random-bias}, where random bias can be viewed as observation error.

%% file: main_results/main_results.tex
This section gives the main PTAS framework for utility configuration.
The algorithm takes as input discrete instances of the problem, that is, instances of utility configuration where for each action $i$ and configuration $j$, the joint distribution $F_{ij}$ over pairs $(u_i^A,u_i^P)$ is given by a list of point masses and their respective probabilities.
The algorithm is parametrized by a fineness parameter $M$, which governs the approximation ratio: the configuration the algorithm computes will be a multiplicative $(1-g(M))$-approximation, where $g(M)\rightarrow 0$ as $M\rightarrow \infty$, with the precise value of $g$ depending on the degree of utility alignment in the instance or family of instances.
For example, for utility configuration problems corresponding to unit-demand pricing, we will have $g(M)=O(\log M/M)$.
Meanwhile, the runtime of the algorithm will grow with $M$; the growth will be exponential in $M$, meaning that the PTAS is not an FPTAS.
(However, as is needed for a PTAS, the runtime is polynomial in the input size for any fixed $M$.)

Before making the PTAS framework fully precise, we give a high-level overview of the main ideas here.
We then prove the necessary preliminary results about the key quantities used in the algorithm in Sections~\ref{sec:main-preprocessing}-\ref{sec:main-objective}, before presenting and analyzing the complete algorithm in Sections~\ref{sec:main-alg} and \ref{sec:main-runtime-analysis}.
We defer discussion of applications to specific problems of interest (and the corresponding analyses of those problems' utility alignment properties) to Section~\ref{sec:applications}.

The algorithm has an outer loop and an inner optimization procedure.
In the outer loop, the algorithm attempts to approximately guess the structure of the distribution of the agent's utility for their favorite action, $\max_i u_i^A$, under the optimal configuration.
For each such guess, the inside optimization solves a dynamic program, which maximizes an approximation to the true objective function, subject to the constraint that the distribution of $\max_i u_i^A$ approximately matches the algorithm's guess.
The algorithm then returns the best solution obtained from the dynamic program, over all the guesses.
We cover each piece in more detail below.

\paragraph{Guessing the agent's utility distribution.}
In the optimal configuration $\config^*$, each action $i$'s agent utility $u_i^A$ is drawn from a distribution $F_{i,C_i^*}$.
The agent chooses the action that yields the highest utility given the realizations $(u_1^A,\ldots, u_n^A)$, yielding utility $\max_i u_i^A$.
For brevity, denote this maximum utility by $u^A=\max_i u_i^A$.
The algorithm's outer loop guesses the $j/M$ quantile of the distribution of $u^A$ for every $j\in\{1,\ldots,M\}$.
More precisely, we guess utilities $\hat u_1\leq\ldots\leq \hat u_M$, which define utility bins $B_j=[\hat u_{j-1},\hat u_j)$ for all $j$, such that for any $j\in \{1,\ldots,M\}$, we have $\Pr[u^A\in B_j]=1/M$.
Assume, for preliminary discussion, an exact partition of the probability space into size $1/M$ bins is possible; with discrete distributions, there will be errors, which we handle in the complete technical exposition.

\paragraph{Optimization procedure.} 
Given a guess of the bins, we now lay out the algorithm's optimization problem.
Rather than maximizing the principal's utility exactly, subject to exact constraints on the probability of the agent's utility falling in each bin $B_j$, the algorithm uses additive estimates.
For the bin probabilities, we consider the conditional probabilities $\qleqj=\Pr[u^A\in B_j~|~u^A\in B_{\leq j}]$, where $B_{\leq j}=\cup_{k\leq j} B_j$.
In the optimal configuration, $\qleqj=1/j$ for each $j\in\{1,\ldots,M\}$.
The algorithm maintains additive estimates of these conditional probabilities.
Given a configuration $\config$, each action has a utility distribution of $u_i^A$ and conditional probability $\qleqij=\Pr[u^A_i\in B_j~|~u^A_i\in B_{\leq j}]$.
We then can consider the additive estimate $\qjest=\sum_{i=1}^n\qleqij$.
Note that if $\qjest$ is close to $1/j$, then each of the individual $\qleqij$s must be small, and hence $\qjest\approx\qleqj$, as first-order terms dominate the calculation of $\qleqj$.
Thus, for $B_j$ corresponding to high quantiles of the agent's utility distribution, requiring $\qjest$ to be close to $1/j$ effectively enforces that the agent's utility distribution is close to that under $\config^*$.

A similar principle applies to the principal's expected utility.
We define $\uleqj=\E[u^P~|~u^A\in B_j]\Pr[u^A\in B_j~|~u^A\in B_{\leq j}]$, where $u^P$ is the principal's utility for the action selected by the agent.
Informally, $\uleqj$ is the principal's expected utility from bin $j$ conditioned on the agent's utility being in bin at most $j$.
Similarly, under configuration $\config$, each action $i$ has a corresponding action-wise $u_{ij}^\leq=\E[u^P_i~|~u^A_i\in B_j]\Pr[u^A_i\in B_j~|~u^A_i\in B_{\leq j}]$.
For similar reasons as with the $\qleqj$s, the additive estimates $\uestj=\sum_{i=1}^nu_{ij}^\leq$ are a good approximation to $u_j^\leq$ when $j$ is large.
Moreover, the contribution of an individual action to the overall objective can be approximated as a weighted sum of its contributions to the additive estimates.
Consequently, maximizing the additive utility estimate subject to the constraint that the additive probability estimates $\qleqj$ be close to $1/j$ yields an optimization problem that is structurally similar to the multidimensional knapsack problem, albeit with both upper and lower bounds on weights, and with multiple choices for each item (instead of ``in'' or ``out'').
As is natural for knapsack-like problems, we then solve the problem approximately using rounding and dynamic programming.

\paragraph{Analysis.}
As discussed, $\qjest$ will be close to $\qleqj$ for high $j$, and $\uestj$ will be an accurate estimate of the principal's utility as well.
With careful accounting, we are able to show that an optimal solution to the optimization above will capture nearly all of the principal's expected utility from high quantiles of the agent's utility distribution.
To account for the low quantiles, then, we rely on utility alignment.
Utility alignment states that under the optimal configuration $\config^*$, these low quantiles of agent utility contribute little to the principal's objective.
In the algorithm, this manifests by down-weighting terms corresponding to low $j$, and shows that the resulting approximate objective value will be close to the true optimal expected principal utility.

\paragraph{Section outline.}
The rest of this section precisely defines and analyzes the estimates above, describes the algorithm in full detail, and analyzes the approximation ratio and runtime of the algorithm.
In Section~\ref{sec:main-preprocessing}, we first define a preprocessing procedure to handle the discrete nature of the utility configuration instances we consider.
This preprocessing step will ensure that it is possible to define utility bins that nearly evenly divide the agent's probability space.
Sections~\ref{sec:main-prob-estimates} and \ref{sec:main-util-estimates} define the additive estimates for the $\qleqj$s and $\uleqj$s, and bound their errors.
Section~\ref{sec:main-objective} then defines the objective function for the algorithm's optimization problem, which approximates the principal's true utility.
Finally, we give the full algorithm in Section~\ref{sec:main-alg} and analyze its runtime in Section~\ref{sec:main-runtime-analysis}.

\input{main_results/preprocessing}

\input{main_results/prob_estimates}

\input{main_results/util_estimates}

\input{main_results/objective}

\input{main_results/algorithm}

\input{main_results/runtime_analysis}

%% file: main_results/preprocessing.tex
\subsection{Preprocessing}\label{sec:main-preprocessing}
The formulation of utility configuration given in Section \ref{sec:intro-utility-configuration} is for a discrete version of the problem.
Recall that in this discrete formulation, action $i$ has $m$ possible utility configurations.
Configuration $j$ for action $i$ is defined by a joint distribution $F_{ij}$ over pairs $(u_i^A,u_i^P)$ of principal utilities, given as a list of $K_{ij}$ triples $(u_{ij}^{A,k},u_{ij}^{P,k},p_{ij}^k)$, where $u_{ij}^{A,k}$ (resp.\ $u_{ij}^{P,k}$) are the agent (resp.\ principal) utilities for this mass point, which is realized with probability $p_{ij}^k$.
In this discrete formulation, probabilities $p_{ij}^k$ are arbitrary numbers in the range $(0,1]$.

The first step of the PTAS is to guess $M$ agent utilities $-\infty=\hat u_0\leq \ldots \leq \hat u_M$ such that under the optimal configuration, the agent's utility is roughly equally likely to fall in each of $M$ bins $B_j=[\hat u_{j-1},\hat u_{j})$ for $j\in\{1,\ldots,M\}$ (with $B_M=[\hat u_{M-1},\hat u_{M}]$).
With general probability masses $p_{ij}^k\in(0,1]$ for utilities, this may not be possible, as large probability masses cause the agent's utility to take particular values with probability significantly greater than the $1/M$ mass allowed in each bin.
To enable us to partition the agent's utility distribution anyway, we now define a preprocessing procedure for breaking larger point masses into smaller masses of size at most $1/M^2$.
In the process, we will perturb agent utilities by a small amount. These perturbations will cause each of the smaller masses we construct to yield a distinct utility for the agent, allowing us to find the desired partition of the agent's utility space.
Meanwhile, these perturbations will not change the expected principal utility from any configuration of the actions.
The preprocessing algorithm is given in Algorithm~\ref{alg:preprocessing}, and explained informally below.

Informally, the preprocessing algorithm splits each point mass into smaller point masses with the same principal utility.
Given a point mass $(u_{ij}^{A,k},u_{ij}^{P,k},p_{ij}^k)$ for configuration $j$ of action $i$, we split it into $\lfloor p_{ij}^k/M^2 \rfloor$ point masses of size $1/M^2$, plus one more of size $p_{ij}^k-\lfloor p_{ij}^k/M^2 \rfloor$ to capture the remaining mass.
The principal utilities for these new point masses will be $u_{ij}^{P,k}$, unchanged from the original.
To induce the agent utilities of each point mass to be distinct, we perturb each new point mass by $\delta\ell$ for some extremely small $\delta>0$, and with $\ell$ assuming a distinct value for each new point mass.
By splitting point masses starting with those with the lowest principal utility, the perturbed utilities enforce tiebreaking in the principal's favor.
We note that it is possible to select a $\delta$ with bit complexity polynomial in the original input size and such that the agent's choice between any two realizations of any two actions with distinct agent utilities is the same as it was before preprocessing; consequently, the principal's utility for any configuration of actions is the same before and after preprocessing.

\begin{algorithm}
\label{alg:preprocessing}
\SetKwInOut{Input}{Input}\SetKwInOut{Output}{Output}
\Input{Distributions $F_{ij}$ for each $i\in[n]$, $j\in[m]$, perturbation $\delta>0$.}
\Output{New distributions $\tilde F_{ij}$ for each $i\in[n]$, $j\in[m]$.}
\BlankLine
\BlankLine
\tcc{Make full sorted list of point masses.}
$L$ $\leftarrow$ empty list\;
\For{$i\in[n]$, $j\in[m]$, $k\in[K_{ij}]$}
    {Append $(u_{ij}^{A,k},u_{ij}^{P,k},p_{ij}^k)$ to $L$\;}
Sort $L$ in non-decreasing order by $u_{ij}^{P,k}$\;

\BlankLine
\BlankLine
\tcc{Empty list for each new distribution.}
\For{$i\in[n]$, $j\in [m]$}{$\tilde F_{ij}$ $\leftarrow$ empty list\;}

\BlankLine
\BlankLine
\tcc{Split and perturb point masses. Append to appropriate distribution.}
$\ell \leftarrow 0$\;
\For{$(u_{ij}^{A,k},u_{ij}^{P,k},p_{ij}^k)$ in $L$}
    {
    \For{$k'\in\{1,\ldots,\lfloor p_{ij}^k/M^2 \rfloor\}$}
        {Append $(u_{ij}^{A,k}+\delta\ell,u_{ij}^{P,k},1/M^2)$ to $\tilde F_{ij}$\tcp*{Add whole increments of $1/M^2$.}
        $\ell\leftarrow\ell+1$\;}
    Append $(u_{ij}^{A,k}+\delta\ell,u_{ij}^{P,k},p_{ij}^k-\lfloor p_{ij}^k/M^2 \rfloor)$ to $\tilde F_{ij}$\tcp*{Add remainder of point mass.}
    $\ell\leftarrow\ell+1$\;
    }
Return $\tilde F_{ij}$ for each $i\in[n]$, $j\in [m]$\;
\caption{Preprocessing Algorithm}
\end{algorithm}

For the remainder of the analysis, we assume that every realization of every action has a distinct agent utility and that each realization has probability at most $1/M^2$.
One consequence of these facts is that the distribution of agent utilities in any solution, including the optimal solution, can now be split into bins of nearly equal probability, stated formally as follows.

\begin{definition}[Utility bins]\label{def:bins}
     Given optimal configuration $\config^*$ after preprocessing, let $\hat u_0=-\infty$, and $\hat u_j$ be the smallest agent utility such that $\prob[\optua\leq \hat u_j]\geq j/M$, for all $j\in\{1,\ldots,M\}$.
     Given the $\hat u_j$s, we may define \emph{utility bins} $B_j=[\hat u_{j-1},\hat u_j)$ for $j\in\{1,\ldots,M-1\}$ and $B_M=[\hat u_{M-1},\hat u_M]$ and \emph{cumulative bins} $B_{\leq j}=\cup_{k\leq j}B_k$.
\end{definition}

After the preprocessing step, the bins constructed by Definition~\ref{def:bins} divide the probability space of the agent's utility distribution nearly equally: Lemma~\ref{lem:even} shows that the agent's utility falls into each of the $M$ bins with probability approximately $1/M$.

\begin{lemma}\label{lem:even}
    The bins of Definition~\ref{def:bins} satisfy:
    \begin{equation*}
        \Pr[\optua\in B_j]\in\left[\frac{1}{M}-\frac{1}{M^2},\frac{1}{M}+\frac{1}{M^2}\right].
    \end{equation*}
\end{lemma}
\begin{proof}
    The cumulative bin probabilities satisfy $\prob[\optua\in \bucket_j]=\prob[\optua\in \bucket_{\leq j}]-\prob[\optua\in \bucket_{\leq j-1}]$.
    Since $\prob[\optua\in \bucket_{\leq j}]\in[j/M,j/M+1/M^2]$ and $\prob[\optua\in \bucket_{\leq j-1}]\in [(j-1)/M,(j-1)/M+1/M^2]$, we can apply these upper and lower bounds:
    \begin{equation*}
        \prob[\optua\in \bucket_{j}]\geq j/M-((j-1)/M+1/M^2)=1/M-1/M^2
    \end{equation*}
    \begin{equation*}
        \prob[\optua\in \bucket_{j}]\leq j/M+1/M^2-(j-1)/M=1/M+1/M^2.
    \end{equation*}
    This implies the lemma.
\end{proof}

Many of the subsequent analyses of the optimal configuration rely on conditional probabilities of the form $\Pr[\optua\in B_k~|~\optua\in B_{\leq j}]$.
The following lemma characterizes these probabilities.
Since the bins nearly evenly divide the probability space, we should expect these probabilities to be roughly $1/j$.
Especially important is the case where $k=j$, as $\qleqj=\bucketprobopt$ is a key quantity in the structure of the PTAS.
The Lemma~\ref{lem:bins} below gives especially tight bounds for these probabilities.

\begin{lemma}\label{lem:bins}
    For any $k\leq j$, the bins of Definition~\ref{def:bins} satisfy
    \begin{equation}\label{eq:kinj}
        \Pr[\optua\in B_k~|~\optua\in B_{\leq j}]\in \left[\frac{1-1/M}{j+1/M},\frac{1+1/M}{j}\right].
    \end{equation}
    Moreover, for $k=j$ the following tighter bound holds:
    \begin{equation}\label{eq:jinj}
        \qleqj = \bucketprobopt\in \left[\frac{1-1/M}{j},\frac{1+1/M}{j+1/M}\right].
    \end{equation}
\end{lemma}
\begin{proof}
    The definition of the bins implies
    \begin{equation}\label{eq:probexpression}
        \Pr[\optua\in B_k~|~\optua\in B_{\leq j}]=\frac{\prob[\optua \in \bucket_{\leq {k}}]-\prob[\optua \in \bucket_{\leq {k-1}}]}{\prob[\optua \in \bucket_{\leq {j}}]}.
    \end{equation}
    To prove (\ref{eq:kinj}), note that $\prob[\optua \in \bucket_{\leq {k}}]\in[k/M,k/M+1/M^2]$, $\prob[\optua \in \bucket_{\leq {k-1}}]\in[(k-1)/M,(k-1)/M+1/M^2]$, and $\prob[\optua \in \bucket_{\leq {j}}]\in[j/M,j/M+1/M^2]$.
    Plugging in the appropriate upper and lower bounds and simplifying yields the quantities in (\ref{eq:kinj}).

    To prove (\ref{eq:jinj}), note that when $k=j$, the term $\prob[\optua \in \bucket_{\leq {j}}]$ appears in both the numerator and denominator of the right-hand side of (\ref{eq:probexpression}).
    This function is increasing in $\prob[\optua \in \bucket_{\leq {j}}]$ and decreasing in $\prob[\optua \in \bucket_{\leq {j-1}}]$.
    Applying the bounds $\prob[\optua\in \bucket_{\leq j}]\in[j/M,j/M+1/M^2]$ and $\prob[\optua\in \bucket_{\leq j-1}]\in [(j-1)/M,(j-1)/M+1/M^2]$ and simplifying, we obtain the desired inequality.
\end{proof}

%% file: main_results/prob_estimates.tex
\subsection{Probability Estimates}\label{sec:main-prob-estimates}

The high-level objective of the PTAS is to construct a configuration with approximately the same distribution for the agent's utility from their favorite action as the optimal solution.
To maintain this approximation, the algorithm aims to match the probabilities that the agent's utility falls into the bins $B_j$ defined as in Definition~\ref{lem:bins}.
More specifically, the algorithm tries to match the conditional probabilities 
$\bucketprobalg$ over the cumulative bins to the optimal conditional probabilities, $\bucketprobopt$.
Rather than exactly track these conditional probabilities, the algorithm will maintain estimates, obtained in two steps, by rounding and then summing the action-wise probabilities $\bucketprobi$.
In more detail, we will analyze the relationship between the following quantities, which may be defined for any configuration $\config$:
\begin{itemize}
    \item Conditional bin probabilities: $\qleqj = \bucketprob$.
    \item Action-wise conditional bin probabilities: $\qleqij = \bucketprobi$.
    \item Rounded action-wise probabilities: $\qleqbarij = \lfloor\qleqij \cdot M^2n\rfloor / M^2n$.
    \item Additive estimates: $\qjest = \sum_{i=1}^n \qleqij$.
    \item Rounded additive estimates $\qjbarest= \sum_{i=1}^n \qleqbarij$.
\end{itemize}

We give three analyses of these quantities.
First, we give a preliminary lemma comparing $\qleqj$ and $\qjest$, which holds for any configuration.
Next, we incorporate the rounding into analyses of the optimal configuration.
We show that in the optimal configuration, $\qjbarest$ will be close to $1/j$, with the estimation error being large for small $j$ corresponding to low-utility buckets, and small for high-utility buckets with large $j$.
(Note that if the distributions were continuous, the true values of $\bucketprobopt$ would be $1/j$ for all $j$.)
Finally, we show that any other solution for which the rounded additive estimates $\qjbarest$ fall within these same bounds, the conditional bin probabilities $\qleqj$ are also close to $1/j$.
This will imply that tracking the rounded additive estimates $\qjbarest$ will achieve the desired goal of approximating the structure of the optimal configuration's distribution of agent utility.

The first task is a general analysis of the unrounded additive estimates.
That is, for an arbitrary configuration $\config$, we compare $\qleqj=\bucketConditionalsExpr{}{j}$ to the additive estimate $\qjest=\sum_{i=1}^n \bucketConditionalsExpr{i}{j}$, where $\utilAgent{}$ is the agent's utility from their favorite action under $\config$.
The following lemma gives the approximation:
\begin{lemma}\label{lemma:prob-estimate}
    For any configuration, the corresponding additive estimates satisfy:
    \begin{equation*}
    (1-\qleqj)\qjest \leq \qleqj \leq \qjest.
    \end{equation*}
\end{lemma}
\begin{proof}
    The rightmost inequality ($\qleqj \leq \qjest$) follows from the union bound.
    The leftmost inequality is proved by the following sequence of inequalities, explained after their statement:
        \begin{align}
        \qleqj & = \bucketConditionalsExpr{}{j}\notag\\
        & \geq \sum\nolimits_{i=1}^n \Pr[\utilAgent{i'} \notin \bucket_j, \forall i' \neq i~|~\utilAgent{} \in \bucket_{\leq j}]\Pr[\utilAgent{i} \in \bucket_j~|~\utilAgent{} \in \bucket_{\leq j}]\label{eq:probest1}\\
        &\geq \sum\nolimits_{i=1}^n \Pr[\utilAgent{} \notin \bucket_j~|~\utilAgent{} \in \bucket_{\leq j}]\Pr[\utilAgent{i} \in \bucket_j~|~\utilAgent{i} \in \bucket_{\leq j}]\label{eq:probest2}\\
        &= (1-\bucketConditionals_j) \sum\nolimits_{i=1}^n \Pr[\utilAgent{i} \in \bucket_j~|~\utilAgent{i} \in \bucket_{\leq j}]\notag\\
        &=(1-\qleqj)\qjest\notag.
    \end{align}
    Equation (\ref{eq:probest1}) follows from the fact that a sufficient condition for action $i$ to be chosen is that all other actions fall in bins with lower utility, and from the independence of the actions.
    Equation (\ref{eq:probest2}) follows from the fact that the event $\utilAgent{}\notin \bucket_j$ is contained in the event $\utilAgent{i'} \notin \bucket_j, \forall i' \neq i$.
\end{proof}

We now analyze the relationship between the rounded estimates $\qleqbarj$ and the true probability $\qleqj$ for the optimal configuration $\config^*$.
If the utility distributions were continuous, then $\qleqj$ would be $1/j$ for all $j$.
The following lemma shows that even with discrete distributions and rounding, the additive estimates are close.
In particular, we obtain low error when $j$ is high, while for constant $j$, the error might be large.

\begin{lemma}\label{lem:optprobs}
    The rounded additive estimates for the optimal configuration satisfy:
    \begin{equation}\label{eq:qbarestjbounds}
        \qjbarest\in \left[\frac{1-1/M}{j}-\frac{1}{M^2},\frac{1+1/M}{j-1}\right].
    \end{equation}
\end{lemma}
\begin{proof}
    We first lower bound $\qjbarest$ with the following sequence of inequalities, explained after their statement:
    \begin{align}
        \qjbarest&\geq \qjest-n\cdot \frac{1}{M^2n}\label{eq:optprobs1}\\
        &=\qjest-\frac{1}{M^2}\notag\\
        &\geq \qleqj-\frac{1}{M^2}\label{eq:optprobs2}\\
        &\geq \frac{1-1/M}{j}-\frac{1}{M^2}.\label{eq:optprobs3}
    \end{align}
Line (\ref{eq:optprobs1}) follows from the fact that there are $n$ actions and for each action $i$, $\qleqij-\qleqbarij\leq 1/(M^2n)$, and line (\ref{eq:optprobs2}) follows from Lemma~\ref{lemma:prob-estimate}.
Line (\ref{eq:optprobs3}) then follows from Lemma~\ref{lem:bins}.

Similarly, we may upper bound $\qjbarest$ as follows:
\begin{align}
    \qjbarest   & \leq \qjest\label{eq:optprobs4}\\
                & \leq \frac{\qleqj}{1-\qleqj}\label{eq:optprobs5}\\
                & \leq \frac{1+1/M}{j-1}\label{eq:optprobs6}.
\end{align}
Line (\ref{eq:optprobs4}) follows because the rounded estimates are obtained by rounding down.
Line (\ref{eq:optprobs5}) follows from Lemma~\ref{lemma:prob-estimate}.
Finally, line (\ref{eq:optprobs6}) from noting that the expression in (\ref{eq:optprobs5}) is increasing in $\qleqj$, plugging in the upper bound from Lemma~\ref{lem:bins}, and simplifying.
\end{proof}

Lemma~\ref{lem:optprobs} implies that there exists at least one configuration, namely $\config^*$, which induces $\qjbarest$ satisfying equation (\ref{eq:qbarestjbounds}).
The PTAS will optimize an estimate of the principal's expected utility over configurations satisfying these bounds.
We next show that any configuration satisfying these bounds also induces values of $\qleqj=\bucketConditionalsExpr{}{j}$ which are close to the $1/j$.

\begin{lemma}\label{lem:esttotrue}
    For any configuration $\config$ satisfying equation (\ref{eq:qbarestjbounds}) for all $j\in\{1,\ldots,M\}$, the resulting conditional probabilities satisfy:
    \begin{equation*}
        \qleqj\in \left[\frac{1-2/M}{j+1-2/M},\frac{1+1/M}{j-1}+1/M^2\right].
    \end{equation*}
\end{lemma}
\begin{proof}
    We first lower bound $\qleqj$.
    The following sequence of inequalities is explained after its statement:
    \begin{align}
        \qleqj  &\geq \frac{\qjest}{1+\qjest}\label{eq:esttotrue1}\\
                &\geq \frac{\qjbarest}{1+\qjbarest}\label{eq:esttotrue2}\\
                &\geq \frac{\frac{1-1/M}{j}-1/M^2}{1+\frac{1-1/M}{j}-1/M^2}\label{eq:esttotrue3}\\
                &= \frac{1-1/M-j/M^2}{j+1-1/M-j/M^2}\notag\\
                &\geq \frac{1-2/M}{j+1-2/M}.\label{eq:esttotrue4}
    \end{align}
    Line (\ref{eq:esttotrue1}) follows from rearranging the bound from Lemma~\ref{lemma:prob-estimate}.
    Line (\ref{eq:esttotrue2}) holds because the right-hand side of line (\ref{eq:esttotrue1}) is an increasing function of $\qjest$, and so we can obtain a lower bound by noting that $\qjbarest\leq \qjest$.
    Line (\ref{eq:esttotrue3}) then follows from plugging in the lower bound from equation (\ref{eq:qbarestjbounds}).
    Finally, line (\ref{eq:esttotrue4}) holds because $j\leq M$.
    
    Next, we can upper bound $\qleqj$.
    Again, we explain the inequalities after their statement.
    \begin{align}
        \qleqj  &\leq \qjest\label{eq:esttotrue5}\\
                &\leq \qjbarest + 1/M^2\label{eq:esttotrue6}\\
                &\leq \frac{1+1/M}{j-1}+1/M^2.\notag
    \end{align}
    Line (\ref{eq:esttotrue5}) follows from Lemma~\ref{lemma:prob-estimate}. Line (\ref{eq:esttotrue6}) follows from the fact that $\qjest\leq \qjbarest+1/M^2$ due to the rounding increment. 
    This implies the lemma.
\end{proof}

%% file: main_results/util_estimates.tex
\subsection{Utility Estimates}\label{sec:main-util-estimates}

The previous section analyzed additive estimates of conditional bin probabilities, which will form the constraints of the dynamic program at the heart of the PTAS.
The next step is to define and analyze additive estimates of the principal's utility.
Specifically, given an arbitrary configuration $\config$, we consider the conditional contributions to the principal's utility defined as $$\uleqj=\bucketExpectationExpr{}{j} \bucketConditionalsExpr{}{j}.$$
We will approximate this with the additive estimate $$\uestj=\sum_{i=1}^n\bucketExpectationExpr{i}{j} \bucketConditionalsExpr{i}{j}.$$
To analyze the PTAS, we will combine these bin-wise estimates into a complete estimate of the principal's expected utility.

The main lemma of this section gives a general approximation relationship which depends only on the conditional bin probability $\qleqj=\bucketConditionalsExpr{}{j}$, stated below.
\begin{lemma}\label{lem:utilapprox}
    For any configuration inducing conditional bin probabilities $\qleqj=\bucketConditionalsExpr{}{j}$, the following holds:
    \begin{equation*}
        (1-\qleqj)\uestj\leq\uleqj\leq \uestj.
    \end{equation*}
\end{lemma}
\begin{proof}
    Let $i^*$ be a random variable denoting the agent's favorite action in $\config$.
    We use the $\I[A]$ to denote the indicator that event $A$ occurs. 
    Spelling out the dependence of the principal's utility on $i^*$, the upper bound follows from the reasoning below:
    \begin{align}
        \uleqj  &= \bucketExpectationExpr{i^*}{j}\bucketConditionalsExpr{i^*}{j}\label{eq:utilub1}\\
                &=\E[u^P_{i^*}\mathbb I[u^A_{i^*}\in B_j]~|~u^A_{i^*}\in B_{\leq j}]\label{eq:utilub2}\\
                &\leq \sum\nolimits_{i=1}^n \E[u^P_i\mathbb I[u^A_i\in 
                B_j]~|~u^A_{i^*}\in B_{\leq j}]\label{eq:utilub3}\\
                &\leq \sum\nolimits_{i=1}^n \E[u^P_i\mathbb I[u^A_i\in 
                B_j]~|~u_1^A\in B_{\leq j},\ldots,u_n^A\in B_{\leq j}]\label{eq:utilub4}\\
                &=\sum\nolimits_{i=1}^n \E[u^P_i\mathbb I[u^A_i\in 
                B_j]~|~u_i^A\in B_{\leq j}]\label{eq:utilub5}\\
                &=\sum\nolimits_{i=1}^n \bucketExpectationExpr{i}{j} \bucketConditionalsExpr{i}{j}\label{eq:utilub6}\\
                &=\uestj\notag.
    \end{align}
    Line (\ref{eq:utilub1}) follows from the definition of $\uleqj$. Lines (\ref{eq:utilub2}) and (\ref{eq:utilub6}) follow from basic probability.
    Line (\ref{eq:utilub3}) follows from the fact that the previous line counts only contributions from $i^*$, instead of those from all actions.
    Line (\ref{eq:utilub4}) holds by the definition of $i^*$, and line (\ref{eq:utilub5}) from the independence of the actions.
    This proves the upper bound.

    Meanwhile, the lower bound follows from the following analysis:
    \begin{align}
        \uleqj&=\bucketExpectationExpr{i^*}{j} \bucketConditionalsExpr{i^*}{j}\notag\\
        &\geq \sum\nolimits_{i=1}^n \E[\utilPrincipal{i}~|~\utilAgent{i} \in \bucket_j]\Pr[\utilAgent{i'} \notin \bucket_j, \forall i' \neq i~|~\utilAgent{\agentPreferredOption}\in \bucket_{\leq j}] \Pr[\utilAgent{i} \in \bucket_j~|~\utilAgent{\agentPreferredOption} \in \bucket_{\leq j}]\label{eq:utillb1}\\
        &\geq \sum\nolimits_{i=1}^n \E[\utilPrincipal{i}~|~\utilAgent{i} \in \bucket_j]\Pr[\utilAgent{\agentPreferredOption} \notin \bucket_j~|~\utilAgent{\agentPreferredOption} \in \bucket_{\leq j}]\Pr[\utilAgent{i} \in \bucket_j~|~\utilAgent{\agentPreferredOption} \in \bucket_{\leq j}]\label{eq:utillb2}\\
        &= \sum\nolimits_{i=1}^n \E[\utilPrincipal{i}~|~\utilAgent{i} \in \bucket_j]\Pr[\utilAgent{\agentPreferredOption} \notin \bucket_j~|~\utilAgent{\agentPreferredOption} \in \bucket_{\leq j}]\Pr[\utilAgent{i} \in \bucket_j~|~\utilAgent{i} \in \bucket_{\leq j}]\label{eq:utillb3}\\
        &= (1-\bucketConditionals_j) \sum\nolimits_{i=1}^n\E[\utilPrincipal{i}~|~\utilAgent{i} \in \bucket_j]\Pr[\utilAgent{i} \in \bucket_j~|~\utilAgent{i} \in \bucket_{\leq j}]\label{eq:utillb4}\\
        &=(1-\qleqj)\uestj.\notag
    \end{align}
    Line (\ref{eq:utillb1}) follows from the fact that a sufficient condition for action $i$ to be selected is that $u_i^A\in B_j$ while all other actions $i'$ have $u_{i'}^A\in B_{\leq j-1}$, along with independence across actions.
    Line (\ref{eq:utillb2}) comes from the fact that $\Pr[\utilAgent{i'} \notin \bucket_j, \forall i' \neq i~|~\utilAgent{\agentPreferredOption}\in \bucket_{\leq j}]\geq \Pr[\utilAgent{\agentPreferredOption} \notin \bucket_j~|~\utilAgent{\agentPreferredOption} \in \bucket_{\leq j}]$ by the definition of $i^*$.
    Line (\ref{eq:utillb3}) follows from independence across actions.
    Finally, line (\ref{eq:utillb4}) follows from the definition of $\qleqj$.
    This proves the lower bound and thus the lemma.
\end{proof}

%% file: main_results/objective.tex
\subsection{Approximate Objective Function} \label{sec:main-objective}

We now define the objective function for the dynamic program at the heart of the PTAS and analyze its relationship to the true objective.
To define the objective function, we consider the following weighted sum of the utility estimates from Section~\ref{sec:main-util-estimates}:
\begin{equation*}
    \objective(\config)=\sum_{j=6}^M \coef_j\uestj,
\end{equation*}
where $\uestj$ implicitly depends on $\config$ for all $j$, and where $\coef_j=\frac{j-5}{M-1}$.
We give two analyses of $\objective$.
First, we show that for any configuration $\config$ where $\qleqj$ satisfies the bounds of Lemma~\ref{lem:esttotrue} for all $j\in\{1,\ldots,M\}$, $\objective(\config)$ is a lower bound on the principal's true utility under that same configuration.
This will serve to lower-bound the objective value of the configurations computed by the PTAS.
We present this analysis in Section~\ref{sec:main-lbalg}.
Second, we show that for the {\em optimal configuration $\config^*$ in particular}, the approximate objective $\objective(\config^*)$ is also an approximate upper bound.
This step relies critically on the utility alignment property (and is the only place in the analysis where the property comes into play).
We give this analysis in Section~\ref{sec:main-ubopt}.
After proving these two lemmas, the stage will be set for the rigorous description of the full PTAS, as well as the final pieces of its analysis, presented in Section~\ref{sec:main-alg}.

\subsubsection{Lower Bounding PTAS Solutions}
\label{sec:main-lbalg}

This section analyzes configurations such that for all $j$, $\qleqj$ satisfies the bounds of Lemma~\ref{lem:esttotrue}, restated below
\begin{equation}\label{eq:esttotruerestate}
    \qleqj\in \left[\frac{1-2/M}{j+1-2/M},\frac{1+1/M}{j-1}+1/M^2\right].
\end{equation}
The result is the following:
\begin{lemma}\label{lem:alglb}
    Let $\config$ be a configuration such that $\qleqj$ satisfies the bounds of Lemma~\ref{lem:esttotrue}. Then the principal's expected utility from $\config$ is at least $\objective(\config)$.
\end{lemma}
\begin{proof}
    We lower bound the principal's expected utility by the following sequence of relations, justified after their statement:
    \begin{align}
        \principalExpectation{} &=\sum\nolimits_{j=1}^M\E[u^P~|~u^A\in \bucket_j]\Pr[u^A\in \bucket_j]\label{eq:alglb1}\\
        &\geq\sum\nolimits_{j=6}^M\E[u^P~|~u^A\in \bucket_j]\Pr[u^A\in \bucket_j]\label{eq:alglb1.5}\\
        &=\sum\nolimits_{j=6}^M\E[u^P~|~u^A\in \bucket_j]\Pr[u^A\in \bucket_j~|~u^A\in \bucket_{\leq j}]\frac{\Pr[u^A\in \bucket_j]}{\Pr[u^A\in \bucket_j~|~u^A\in \bucket_{\leq j}]}\notag\\
        &=\sum\nolimits_{j=6}^M\uleqj\frac{\Pr[u^A\in \bucket_j]}{\qleqj}\label{eq:alglb2}\\
        &\geq \sum\nolimits_{j=6}^M(1-\qleqj)\uestj\frac{\Pr[u^A\in \bucket_j]}{\qleqj}\label{eq:alglb3}\\
        &\geq \sum\nolimits_{j=6}^M(1-\qleqj)\uestj\frac{\qleqj\prod_{k>j}(1-q_k^\leq)}{\qleqj}\label{eq:alglb4}\\
        &=\sum\nolimits_{j=6}^M\prod_{k\geq j}(1-q_k^\leq)\uestj\notag\\
        &\geq \sum\nolimits_{j=6}^M\coef_j\uestj\label{eq:alglb5}\\
        &=\objective(\config).\notag
    \end{align}
    Line (\ref{eq:alglb1}) follows from the law of total expectation, and line (\ref{eq:alglb1.5}) from the fact that the principal's utilities are nonnegative. Line (\ref{eq:alglb2}) follows from the definitions of $\qleqj$ and $\uleqj$.
    Line (\ref{eq:alglb3}) holds because of Lemma~\ref{lem:utilapprox}. 
    Line (\ref{eq:alglb4}) follows by Lemma~\ref{lem:binprobs}, proved below, which states that  $\Pr[u^A\in \bucket_j]=\qleqj\prod_{k>j}(1-q_k^\leq)$.
    Finally, Line (\ref{eq:alglb5}) holds because of Lemma~\ref{lem:product}, also proved below, which states that when $\qleqj$ satisfies (\ref{eq:esttotruerestate}) for all $j$, $\prod_{k\geq j}(1-q_k^\leq)\geq \frac{j-5}{M-1}=\coef_j$ for $j\geq 6$.
\end{proof}

We conclude by stating and proving the technical lemmas required for Lemma~\ref{lem:alglb}.
\begin{lemma}\label{lem:binprobs}
    For all $j\in\{1,\ldots,M\}$, the bin probabilities satisfy:
    \begin{equation*}
        \Pr[u^A\in \bucket_j]=\qleqj\prod_{k>j}(1-q_k^\leq)=\bucketConditionalsExpr{}{j}\prod_{k>j}\Pr[u^A\in\bucket_{\leq k-1}~|~u^A\in\bucket_{\leq k}].
    \end{equation*}
\end{lemma}
\begin{proof}
    The right-hand expression follows from the middle one by the definition of the $\qleqj$s.
    To see that the right-hand expression is equal to $\Pr[u^A\in B_j]$, we use the law of iterated expectation. 
    Consider the following sequential process for realizing $u^A$: first determine whether $u^A\in B_M$, which occurs with probability equal to $\Pr[u^A\in B_M]=\Pr[u^A\in B_M~|~u^A\in B{\leq M}]=q_M^\leq$.
    If this does not occur, then determine whether $u^A\in B_{M-1}$, which, conditioned on reaching this point, occurs with probability $\Pr[u^A\in B_{M-1}~|~u^A\in B_{\leq M-1}]=q_{M-1}^\leq$, and so on for each box from $M-2$ down to $j$.
    For $u^A\in B_j$, then, $u^A$ must have failed to lie in $B_k$ for all $k>j$, which happens with overall probability $\prod_{k>j}(1-q_k^\leq)$.
    The probability of finally falling in $B_j$ given these previous failures is then $\Pr[u^A\in B_j~|~u^A\in B_{\leq j}]=\qleqj$, for an overall probability of $\qleqj\prod_{k>j}(1-q_k^\leq)$.
\end{proof}

\begin{lemma}\label{lem:product}
    For any configuration inducing $\qleqj$s satisfying (\ref{eq:esttotruerestate}) for all $j\in\{1,\ldots,M\}$, it holds that $\prod_{k\geq j}(1-q_k^\leq)\geq \frac{j-5}{M-1}$ for $j\geq 6$.
\end{lemma}
\begin{proof}
	At a high level, the result will follow from the fact that $q_k^{\leq}\approx 1-1/k=(k-1)/k$ for all $k\geq j$.
	Were this to hold exactly, then the product would telescope, yielding $(j-1)/M$.
	The crux of the analysis, then, is showing that multiplying the errors introduced in the $q_k^\leq$s does not cause significant loss.
	Formally, the lemma follows from the inequalities below, explained after their statement.
	\begin{align}
		\prod_{k\geq j}(1-q_k^\leq)&\geq \prod_{k\geq j}\left(1- \frac{1+1/M}{k-1}-\frac{1}{M^2}\right)\label{eq:prod1}\\
		&= \prod_{k\geq j}\left(\frac{k-2-1/M}{k-1}-\frac{1}{M^2}\right)\notag\\
		&\geq \prod_{k\geq j}\left(\frac{k-2-1/M}{k-1}-\frac{1/M}{k-1}\right)\label{eq:prod2}\\
		&= \prod_{k\geq j}\frac{k-2-2/M}{k-1}\notag\\
		&=\frac{j-2-2/M}{j-1}\cdot\frac{j-1-2/M}{j}\cdot\ldots\cdot\frac{M-3-2/M}{M-2} \cdot \frac{M-2-2/M}{M-1}\notag\\
		&=\frac{j-2-2/M}{M-1}\prod_{k=j-1}^{M-2}\left(\frac{k-2/M}{k}\right)\label{eq:prod3}\\
		&=\frac{j-2-2/M}{M-1}\prod_{k=j-1}^{M-2}\left(1-\frac{2}{Mk}\right)\notag\\
		&\geq \frac{j-2-2/M}{M-1}\left(1-\sum_{k=j-1}^{M-2}\frac{2}{Mk}\right)\notag\\
		&=\frac{j-2-2/M}{M-1}\left(1-\frac{2}{M}(H_{M-2}-H_{j-2})\right)\notag\\
		&\geq \frac{j-2-2/M}{M-1}\left(1-\frac{2}{M}(\log(M-2)+1-\log(j-2))\right)\label{eq:prod4}\\
		&=\frac{j-2-2/M}{M-1}-\frac{j-2-2/M}{M-1}\cdot\frac{2}{M}\left(\log\left(\frac{M-2}{j-2}\right)+1\right)\notag\\
		&\geq\frac{j-2-2/M}{M-1}-\frac{2}{M-1}\cdot\frac{j-2}{M-2}\left(\log\left(\frac{M-2}{j-2}\right)+1\right)\label{eq:prod5}\\
		&=\frac{j-2-2/M}{M-1}-\frac{2}{M-1}\left(\frac{j-2}{M-2}\log\left(\frac{M-2}{j-2}\right)+\frac{j-2}{M-2}\right)\notag\\
		&\geq\frac{j-2-2/M}{M-1}-\frac{2}{M-1}\label{eq:prod6}\\
		&\geq\frac{j-5}{M-1}.\label{eq:prod7}
	\end{align}
	Line (\ref{eq:prod1}) follows from the lower bound (\ref{eq:esttotruerestate}).
	Line (\ref{eq:prod2}) follows from the fact that $k-1\leq M$.
	Line (\ref{eq:prod3}) follows from rearranging the numerators and denominators of the expanded product on the previous line.
	Line (\ref{eq:prod4}) holds because $H_k\in[\log n,\log n + 1]$.
	Line (\ref{eq:prod5}) follows from rearranging the fractions in the second term, and noting that $j-2-2/M\leq j-2$ on top while $M-2\leq M$ on bottom.
	Finally, line (\ref{eq:prod6}) follows because $x\log x+x\leq 1$ for all $x\in[0,1]$, and line (\ref{eq:prod7}) because $M\geq 2$.
	All other lines follow from algebra or well-known inequalities.
\end{proof}

\subsubsection{Upper Bounding Optimal Solutions}
\label{sec:main-ubopt}

We now analyze the relationship between the true principal utility under the optimal configuration $\config^*$ and the approximate objective value $\objective(\config^*)$.
The approximate objective disproportionately down-weights the utility contribution from lower bins but the error from down-weighting disappears as the index $j$ of the bucket gets closer to $M$, corresponding to stronger quantiles of the agent's utility distribution.
Utility alignment implies that most of the optimal utility comes from these stronger quantiles.
Hence, the approximate objective $\objective(\config^*)$ is not much less than the principal's true utility in $\config^*$.
The upper bound on the optimal utility will depend on the degree of utility alignment in the utility configuration instance.
The following notation will allow us to state the bound.
\begin{definition}\label{def:rs}
    Given an instance of utility configuration satisfying $f(q)$-utility alignment, the {\em utility alignment coefficients} are defined as $r_j=f(j/M)$ for $j\in\{1,\ldots,M\}$.
\end{definition}

Utility alignment implies that for the optimal configuration and any $j\in \{1,\ldots,M\}$, the conditional expectation $\E[\optup~|~\optua\in B_{\leq j}]$ is at most $r_j\E[\optup]$.
We then have the following bound on the optimal utility.
\begin{lemma}\label{lem:optub}
    Let $\config^*$ be the optimal configuration. Then the principal's optimal utility $\E[\optup]$ is approximately upper-bounded by the objective $\objective(\config^*)$, namely:
    \begin{equation}\label{eq:optub}
    	\objective(\config^*)\geq\frac{M-1}{M+1}\left(\frac{M-5}{M-1}-\frac{5}{6}\frac{r_5}{M-1}-\frac{5}{M-1}\sum_{j=6}^M\frac{r_j}{j-1}\right)\E[\optup].
    \end{equation}
\end{lemma}
To parse the complicated bound of (\ref{eq:optub}), note that if the $r_j$s are constant (e.g.\ $r_j=2$ for all $j$), we may conclude that $\objective(\config^*)\geq (1-O(\log M/M))\E[\optup]$.
Similarly non-constant $r_j$s that are smaller than $M/j$ (e.g. $r_j=\sqrt{M/j}$) will yield weaker approximations that still approach $1$ as $M$ grows.
Before proving the lemma, note that all the expressions considered so far for the utility of a configuration have been in terms of the utility conditioned on the agent's utility falling in a particular bin, $\E[u^P~|~u^A\in \bucket_j]$.
Utility configuration is stated in terms of the event that the agent's utility falls into cumulative bins, $\E[u^P~|~u^A\in \bucket_{\leq j}]$.
The following lemma will bridge the gap between these two accounting strategies and allow us to use summation by parts in the proof of Lemma~\ref{lem:optub}.

\begin{lemma}\label{lem:differences}
    In the optimal configuration $\config^*$, for any bin $j$, the principal's expected utility satisfies:
    \begin{equation*}
        \bucketExpectationExpr{\opt}{j}\geq \frac{j+1/M}{1-1/M}\bucketExpectationExpr{\opt}{\leq j}-\frac{j-1}{1+1/M}\bucketExpectationExpr{\opt}{\leq j-1}.
    \end{equation*}
\end{lemma}
\begin{proof}
    We start by rewriting each of the terms on the right-hand side, using the bounds from Lemma~\ref{lem:bins}, which show that $\Pr[\utilAgent{\opt} \in \bucket_k | \utilAgent{\opt} \in \bucket_{\leq j}]\in[\tfrac{1-1/M}{j+1/M},\tfrac{1+1/M}{j}]$ for any $k\leq j$.
    \begin{align*}
        \bucketExpectationExpr{\opt}{\leq j} &= \sum_{k=1}^j \bucketExpectationExpr{\opt}{k} \cdot \Pr[\utilAgent{\opt} \in \bucket_k | \utilAgent{\opt} \in \bucket_{\leq j}]\\
        &\geq\frac{1-1/M}{j+1/M} \sum_{k=1}^j \bucketExpectationExpr{\opt}{k}.
    \end{align*}
    \begin{align*}
        \bucketExpectationExpr{\opt}{\leq j-1} &=  \sum_{k=1}^{j-1} \bucketExpectationExpr{\opt}{k} \cdot \Pr[\utilAgent{\opt} \in \bucket_k | \utilAgent{\opt} \in \bucket_{\leq j-1}]\\
        &\leq \frac{1+1/M}{j-1}\sum_{k=1}^{j-1} \bucketExpectationExpr{\opt}{k}.
    \end{align*}
    Combining these two inequalities completes the proof.
\end{proof}

\begin{proof}[Proof of Lemma~\ref{lem:optub}]
	The proof consists of three main steps.
	First, we note that the objective is defined in terms of the utility estimates $\uestj$, which are lower bounded by the true conditional utilities $\uleqj=\E[\optup~|~\optua\in B_j]\Pr[\optua\in B_j~|~\optua\in B_{\leq j}]$.
	Next, we use Lemma~\ref{lem:differences} to rewrite the conditional utilities in terms of differences of the cumulative utilities $\bucketExpectationExpr{\opt}{\leq j}$.
	Finally, we use summation by parts to rewrite the quantity in a form where we may apply utility alignment, yielding a bound in terms of the true optimal utility $E[\optup]$.
	Formally, we have the following, explained after the statement:
    \begin{align}
    	\objective(\config^*)&=\sum_{j=6}^M\coef_j\uestj\notag\\
        &\geq \sum_{j=6}^M\coef_j\uleqj\label{eq:optub1}\\
        &=\sum_{j=6}^M\coef_j\E[\optup~|~\optua\in B_j]\Pr[\optua\in B_j~|~\optua\in B_{\leq j}]\label{eq:optub2}\\
        &\geq\left(1-\frac{1}{M}\right)\sum_{j=6}^M\frac{\coef_j}{j}\E[\optup~|~\optua\in B_j]\label{eq:optub3}\\
        &\geq\left(1-\frac{1}{M}\right)\sum_{j=6}^M\frac{\coef_j}{j} \Big(\frac{j+1/M}{1-1/M}\bucketExpectationExpr{\opt}{\leq j}\notag\\
    	&\quad\quad\quad\quad -\frac{j-1}{1+1/M}\bucketExpectationExpr{\opt}{\leq j-1}\Big)\label{eq:optub4}\\
    	&\geq\left(1-\frac{1}{M}\right)\sum_{j=6}^M\frac{\coef_j}{j} \Big(\frac{j}{1+1/M}\bucketExpectationExpr{\opt}{\leq j}\notag\\
    	&\quad\quad\quad\quad-\frac{j-1}{1+1/M}\bucketExpectationExpr{\opt}{\leq j-1}\Big)\notag\\
    	&=\frac{M-1}{M+1}\sum_{j=6}^M\frac{\coef_j}{j} \Big(j\bucketExpectationExpr{\opt}{\leq j}-(j-1)\bucketExpectationExpr{\opt}{\leq j-1}\Big)\notag\\
    	&=\frac{M-1}{M+1}\Bigg(c_M\E[\optup]-\frac{c_6}{6}\cdot5\bucketExpectationExpr{\opt}{\leq 5}\notag\\
    	&\quad\quad\quad\quad-\sum_{j=6}^M\left(\frac{c_j}{j}-\frac{c_{j-1}}{j-1}\right)j\cdot\bucketExpectationExpr{\opt}{\leq j}\Bigg)\label{eq:optub5}\\
        &\geq\frac{M-1}{M+1}\left(c_M\E[\optup]-\frac{5}{6}c_6r_5\E[\optup]-\sum_{j=6}^M\left(\frac{c_j}{j}-\frac{c_{j-1}}{j-1}\right)r_j\cdot j\cdot\E[\optup]\right)\label{eq:optub6}\\
        &=\frac{M-1}{M+1}\left(c_M-\frac{5}{6}c_6r_5-\sum_{j=6}^M\left(\frac{c_j}{j}-\frac{c_{j-1}}{j-1}\right)r_j\cdot j\right)\E[\optup]\notag\\
        &=\frac{M-1}{M+1}\left(\frac{M-5}{M-1}-\frac{5}{6}\frac{r_5}{M-1}-\sum_{j=6}^M\left(\frac{j-5}{j(M-1)}-\frac{j-6}{(j-1)(M-1)}\right)r_j\cdot j\right)\E[\optup]\notag\\
        &=\frac{M-1}{M+1}\left(\frac{M-5}{M-1}-\frac{5}{6}\frac{r_5}{M-1}-\frac{1}{M-1}\sum_{j=6}^M\left(\frac{j-5}{j}-\frac{j-6}{j-1}\right)r_j\cdot j\right)\E[\optup]\notag\\
        &=\frac{M-1}{M+1}\left(\frac{M-5}{M-1}-\frac{5}{6}\frac{r_5}{M-1}-\frac{5}{M-1}\sum_{j=6}^M\frac{r_j}{j-1}\right)\E[\optup]\notag.
    \end{align}
Line (\ref{eq:optub1}) follows from Lemma~\ref{lem:utilapprox}, and line (\ref{eq:optub2}) from the definition of the $\uleqj$s.
Line (\ref{eq:optub3}) follows from Lemma~\ref{lem:bins}, which states that $\Pr[\optua\in B_j~|~\optua\in B_{\leq j}]\geq (1-1/M)/j$.
Line (\ref{eq:optub4}) then follows from Lemma~\ref{lem:differences}.
In line (\ref{eq:optub5}), we apply summation by parts, and in line (\ref{eq:optub6}), we apply utility alignment.
The remaining lines are algebra.
\end{proof}

%% file: main_results/algorithm.tex
\subsection{PTAS Description and Approximation Analysis} \label{sec:main-alg}

We now present and analyze the full PTAS for utility configuration.
Recall that an input to the utility configuration problem consists of $n$ actions, each of which has $m$ possible configurations, each represented by a discrete distribution over $(u^A_i,u^P_i)$ pairs.
Before executing the algorithm's main loop, we preprocess the discrete distributions to obtain an equivalent instance where each probability mass in each utility distribution has size at most $1/M^2$, and such that agent utilities are distinct across all point masses for all configurations for all actions.
Full details and analysis of the preprocessing algorithm can be found in Section~\ref{sec:main-preprocessing}.

For a constant $M$ (which governs the approximation factor), the main loop of the algorithm iterates through all profiles of bin boundaries $\hat u_1\leq\cdots\leq \hat u_M$, chosen from the union of all supports of all agent utility distributions for all actions and configurations.
Note that the number of such profiles is exponential in $M$, but for constant $M$, it is polynomial in the input size.
These bin boundaries then define bins $B_1,\ldots, B_M$ in agent utility space, which then allow us to define the key quantities analyzed in Sections~\ref{sec:main-prob-estimates}, \ref{sec:main-util-estimates}, and \ref{sec:main-objective}.
These quantities then define the optimization we perform for each profile of bins.

For each choice of $\hat u_1\leq\cdots\leq \hat u_M$, we solve a dynamic program, with the objective of maximizing the approximate objective $\objective$ defined in Section~\ref{sec:main-objective}, subject to constraint that the agent utilities induced by the chosen configuration should be approximately evenly distributed between the buckets $B_1,\ldots,B_M$, which we enforce using the (rounded) additive estimates $\qjbarest$ of $\Pr[u^A\in B_j~|~u^A\in B_{\leq j}]$, defined in Section~\ref{sec:main-prob-estimates}.
Formally, we solve the following problem for each choice of $\hat u_1\leq\cdots\leq \hat u_M$:
\begin{align*}
	&\text{maximize} \quad\quad\quad \objective(\config) &\\
	&\text{subject to}\quad\quad\quad\qjbarest\in \left[\frac{1-1/M}{j}-\frac{1}{M^2},\frac{1+1/M}{j-1}\right]  \quad\quad\quad \forall j \in\{1,\ldots,M\},
\end{align*}
where recall that each $\qjbarest$ depends on the choice of configuration.
We discuss the details of the dynamic programming formulation of this problem below.
While the program above may not be feasible for all choices of $\hat u_1\leq\cdots\leq \hat u_M$, the analysis of the previous sections implies that it will be feasible for at least one choice, corresponding to a correct guess of the bin boundaries in the true optimal configuration.
The overall algorithm computes the expected principal utility from each configuration obtained by solving the program above for each choice of $\hat u_1\leq\cdots\leq \hat u_M$, and returns the best one overall.\footnote{Given a configuration, computing the principal's true expected utility can be done by emulating the approach of \citet{KPT24}: let $a^*(i)$ denote the agent's favorite action among $1,\ldots,i$. Let $\text{Util}_i(u)=\E[u^P_{a^*(i)}~|~u^A_{a^*(i)}=u]$ and $p_i(u)=\Pr[u^A_{a^*(i)}=u]$. Both can be computed by dynamic programming, and the expected principal utility for any configuration can be computed using knowledge of the $\text{Util}_M(u)$s, $p_M(u)$s, and the law of total expectation.}

We next show that for constant $M$, the optimization problem above can be solved in polynomial time via dynamic programming.
Given a configuration $\config$, the objective $\objective(\config)$ is defined as follows, making the dependence on the chosen configuration explicit:
\begin{align*}
	\sum\nolimits_{j=6}^M c_j\uestj&=\sum\nolimits_{j=6}^M c_j\sum\nolimits_{i=1}^n \E_{(u^A_i,u^P_i)\sim F_{i,C_i}}[u^P_i~|~u_i^A\in B_j]\prob_{(u^A_i,u^P_i)\sim F_{i,C_i}}[u_i^A\in B_j~|~u_i^A\in B_{\leq j}]\\
	&=\sum\nolimits_{i=1}^n\sum\nolimits_{j=6}^M c_j \E_{(u^A_i,u^P_i)\sim F_{i,C_i}}[u^P_i~|~u_i^A\in B_j]\prob_{(u^A_i,u^P_i)\sim F_{i,C_i}}[u_i^A\in B_j~|~u_i^A\in B_{\leq j}].
\end{align*}
Hence, if we define 
\begin{equation*}
	\psi_{i}(\configi)=\sum\nolimits_{j=6}^M c_j\E_{(u^A_i,u^P_i)\sim F_{i,C_i}}[u^P_i~|~u_i^A\in B_j]\prob_{(u^A_i,u^P_i)\sim F_{i,C_i}}[u_i^A\in B_j~|~u_i^A\in B_{\leq j}]
\end{equation*}
to be the contribution of action $i$ to $\objective(\config)$, then $\objective(\config)=\sum_i\psi_i(\configi)$.

To define the subproblem and recurrence for the dynamic program, recall the definitions of the probability estimates $\qjbarest$, where we now make dependence on the configuration $\config$ explicit: we have $\qleqij(\configi) = \prob_{(u^A_i,u^P_i)\sim F_{i,\configi}}[u^A_i\in B_j~|~u^A_i\in B_{\leq j}]$,  $\qleqbarij(\configi) = \lfloor\qleqij(\configi)/(M^2n)\rfloor$, and $\qjbarest(\config)= \sum_{i=1}^n \qleqbarij(\configi)$.
From these definitions, each action contributes to $\qjbarest$ in increments of $1/M^2n$.
Thus, given $i\in\{1,\ldots,n\}$ and $\mathbf k=(k_1,\ldots,k_M)\in[M^2n]^M$, define $\opt(i,\mathbf k)$ to be the maximum value of $\sum_{i'=i}^M\psi_{i'}(C_{i'})$ over all configurations of actions $i,\ldots,M$ such that for each $j\in\{1,\ldots,M\}$, it holds that 
\begin{equation*}
	\frac{k_j}{M^2n}+\sum_{i'=i}^M\qleqbarij(C_{i'})\in \left[\frac{1-1/M}{j}-\frac{1}{M^2},\frac{1+1/M}{j-1}\right].
\end{equation*}
Then $\opt(1,\mathbf 0)$ corresponds to the optimal solution to the optimization problem above.
The base cases are $\opt(n+1,\mathbf k)$ for all $\mathbf k\in [M^2n]^M$, with value either $0$ or $-\infty$ depending on whether $\mathbf k/M^2n$ satisfies the program's constraints on $\qjbarest$.
To define the recurrence, consider configuration $\ell$ of action $i$, and let $k_j^{i,\ell}=\qleqbarij(\ell) \cdot M^2n$ be the number of increments action $i$ would contribute to bin $j$ in configuration $\ell$.
Then we have the following recurrence:
\begin{equation*}
    \opt(i,\mathbf k)=\max_{\ell\in \{1,\ldots,m\}}\psi_i(\ell)+\opt(i+1,\mathbf k-(k_1^{i,\ell},\ldots,k_M^{i,\ell})).
\end{equation*}
We analyze the runtime more carefully in Section~\ref{sec:main-runtime-analysis}, but for now note that each $k_j$ can range from $1$ to $M^2n$ and there are at most $m$ choices at each stage of the optimization, which for any constant $M$ will yield a polynomial-time algorithm in $n$, $m$, and the bit complexity of the utilities and probabilities.

We can now analyze the approximation ratio of the algorithm described above.
The approximation ratio will depend on the choice of $M$ and the utility alignment of the utility configuration instance, which we summarize by the coefficients $r_1,\ldots,r_M$ (Definition~\ref{def:rs}).

\begin{theorem}\label{thm:main}
	For an instance of utility configuration with alignment coefficients $r_1,\ldots,r_M$, the algorithm of this section achieves an $\alpha$-approximation to the optimal principal utility, where
    \begin{equation*}
        \alpha=\frac{M-1}{M+1}\left(\frac{M-5}{M-1}-\frac{5}{6}\cdot\frac{r_5}{M-1}-\frac{5}{M-1}\sum_{j=6}^M\frac{r_j}{j-1}\right).
    \end{equation*}
\end{theorem}
\begin{proof}
   Consider the choice of $\hat u_1,\ldots,\hat u_M$ corresponding to the optimal configuration $\config^*$ as in Definition~\ref{def:bins}, and define the bins $B_1,\ldots, B_M$ and $B_{\leq 1},\ldots,B_{\leq M}$ accordingly.
   These, in turn, set the approximate objective function $\objective$, which depends on the bins to be well-defined.
   Let $\config$ denote the configuration computed by the algorithm for this choice of $\hat u_1,\ldots,\hat u_M$.
   By Lemma~\ref{lem:optprobs}, the optimal configuration is feasible for this choice of $\hat u_1,\ldots,\hat u_M$, and hence the algorithm does indeed find some configuration on this iteration.
   Let $\E[u^P_{\text{ALG}}]$ denote the expected principal utility for $\config$.
   Then we have the following:
   \begin{align}
       \E[u^P_{\text{ALG}}]&\geq \objective(\config)\label{eq:a1}\\
                        &\geq \objective(\config^*)\label{eq:a2}\\
                        &\geq \alpha \E[u^A_{\text{OPT}}].\label{eq:a3}
   \end{align}
   Line (\ref{eq:a1}) follows from the fact that the feasibility constraints of the dynamic program imply that $\config$ satisfies the bounds of Lemma~\ref{lem:optprobs} and hence the true $\qleqj$s satisfy the bounds of Lemma~\ref{lem:esttotrue}.
   Lemma~\ref{lem:alglb} implies that $\objective$ is a lower bound on the expected principal utility for such configurations.
   Line (\ref{eq:a2}) follows from Lemma~\ref{lem:optprobs}, which implies that the optimal configuration is feasible for the dynamic program, and therefore can have no higher objective value than $\config$.
   Finally, (\ref{eq:a3}) follows from Lemma~\ref{lem:optub}.
   Since the algorithm of this section returns a configuration that has at least as high expected principal utility as $\config$ does, the result follows.
\end{proof}

We conclude the section by applying Theorem~\ref{thm:main} with the two sets of alignment coefficients that will apply for the applications in this paper, and work out the resulting approximation ratios as a function of $M$.
In each case, we obtain an approximation ratio which goes to $1$ as $M$ grows, implying (given the runtime analysis of the next section) that the algorithm is a PTAS.
\begin{corollary}\label{cor:2}
    For utility configuration instances satisfying $2$-utility alignment, the algorithm of this section is a $(1-O(\log M/M))$-approximation.
\end{corollary}
\begin{proof}
    For $2$-utility aligned instances, we have $r_j=2$ for all $j$.
    Plugging these $r_j$ into the formula for the approximation ratio, we obtain:
    \begin{equation*}
    \frac{M-1}{M+1}\left(\frac{M-5}{M-1}-\frac{5}{6}\cdot\frac{2}{M-1}-\frac{5}{M-1}\sum_{j=6}^M\frac{2}{j-1}\right).
    \end{equation*}
    The final term is the dominant source of loss, and is $O(\log M/M)$, implying the result.
\end{proof}

\begin{corollary}\label{cor:sqrt}
    For utility configuration instances satisfying $1/(c\sqrt q)$-utility alignment for constant $c$ the algorithm of this section is a $(1-O(1/\sqrt{M}))$-approximation.
\end{corollary}
\begin{proof}
    For such instances, we have $r_j=c\sqrt{M/j}$ for all $j$.
    Therefore, the approximation ratio is:
    \begin{align}
        &\frac{M-1}{M+1}\left(\frac{M-5}{M-1}-\frac{5}{6}\cdot\frac{c\sqrt M}{5(M-1)}-\frac{5}{M-1}\sum_{j=6}^M\frac{c\sqrt M}{(j-1)\sqrt j}\right).\notag\\
        &\quad\quad\quad\quad =\frac{M-1}{M+1}\left(\frac{M-5}{M-1}-\frac{5}{6}\cdot\frac{c\sqrt M}{5(M-1)}-\frac{5c\sqrt M}{M-1}\sum_{j=6}^M\frac{1}{(j-1)\sqrt j}\right).\label{eq:sqrtanalysis}
    \end{align}
    The sum in (\ref{eq:sqrtanalysis}) is now $O(1)$.
    The second and third terms in parentheses are now the dominant loss terms, and are both $O(1/\sqrt{M})$.
\end{proof}

%% file: main_results/runtime_analysis.tex
\subsection{Runtime Analysis}\label{sec:main-runtime-analysis}
We conclude the section with a runtime analysis of the algorithm and show that for fixed $M$, the runtime is polynomial in the number of actions $n$, the number of configurations per action $m$, the number of masses $K$ per distribution, and the bit complexity of the utilities and probabilities given as inputs to the problem.

First, consider preprocessing (Section~\ref{sec:main-preprocessing}).
The preprocessing step takes each probability mass, splits it into at most $M^2$ smaller probability masses, and adds a perturbation to each agent's utility in such a way that the perturbations have polynomial bit complexity in the size of the original input. 
This blows up the number of mass points by at most a factor of $M^2$, while keeping the number of actions and configurations per action the same.
Hence, the size of the output to preprocessing is polynomial in the original input size, and in particular has at most $nmKM^2$ different masses across all actions and configurations, each with a different agent utility.

Next, consider the runtime of the main loop.
The algorithm tries each choice of $\hat u_1,\ldots,\hat u_M$, and after preprocessing, there are $nmKM^2$ choices for each $\hat u_j$, leading to at most $(nmKM^2)^M$ choices of $\hat u_1,\ldots,\hat u_M$.
We next consider the relevant aspects of the dynamic program, which is solved once per choice of $\hat u_1,\ldots,\hat u_M$.
It is straightforward to verify that the computation of the $\psi_i$s and $\qleqbarij$s can be done in polynomial time, and is not the bottleneck to the computation.
The dimension of the memo table, meanwhile, is $n(M^2n)^M$, and each entry requires $m$ constant-time operations to fill.
Therefore, the runtime of the dynamic program is $O(mn(M^2n)^M)$.
Since this is run once per loop, we obtain a runtime of $O(mn(M^2n)^M\cdot (nmKM^2)^M)=O(m^{M+1}n^{2M+1}K^MM^{4M})$, which is polynomial when $M$ is constant.

%% file: applications/applications.tex
This section now applies the utility configuration PTAS framework from Section~\ref{sec:main-results} to the problems discussed in the introduction.
For each problem, we first show how to reduce the problem to utility configuration.
We then prove that each problem satisfies utility alignment (Definition \ref{def:alignment}), which we restate below:

\begin{definition*}[Utility Alignment, restated]
    Given a non-increasing function $f$, an instance of utility configuration is {\em$f(q)$-utility aligned} if for any utility threshold $U$ such that in the optimal configuration $\config^*$, $\prob[u^A\leq U]=q$, the principal utilities under $\config^*$ satisfy
    \begin{equation*}
        \E[u^P~|~u^A\leq U]\leq f(q) \E[u^P],
    \end{equation*}
    where $u^A$ and $u^P$ are the agent and principal utilities from the agent's preferred action under $\config^*$.
\end{definition*}

In Sections~\ref{sec:applications-delegation} and \ref{sec:applications-pricing}, we consider the delegation and unit demand pricing problems discussed in the introduction.
For each of these problems, we show $2$-utility alignment.
In Section~\ref{sec:applications-general}, we then present a broad result that implies utility alignment for several similar problems.
The result shows that if utility alignment holds for individual actions, then it also holds for optimal configurations, albeit at a weaker level (though still strong enough to imply a PTAS).
This result implies that the algorithm of the previous section is a PTAS for two variants of the delegated choice problem (one with randomly distributed bias, and another with an outside option), as well as the assortment optimization problem.
We cover these in detail in Sections~\ref{sec:applications-delegation-random-bias}-\ref{sec:applications-assortment}.

\subsection{Delegated Choice}\label{sec:applications-delegation}
\input{applications/delegation}
\subsection{Unit-demand Pricing}\label{sec:applications-pricing}
\input{applications/pricing}
\subsection{Local to Global Utility Alignment}\label{sec:applications-general}
\input{applications/general/general}
\subsubsection{Delegation with Independent Random Biases}\label{sec:applications-delegation-random-bias}
\input{applications/general/delegation_random}
\subsubsection{Delegation with An Outside Option}\label{sec:applications-delegation-outside-option}
\input{applications/general/delegation_outside_option}
\subsubsection{Assortment Optimization}\label{sec:applications-assortment}
\input{applications/general/assortment_optimization}

%% file: applications/delegation.tex
Recall that in the delegated choice problem, there are $n$ actions, each with an associated bias $b_i$ and randomly (independently) distributed nonnegative value $v_i$ with distribution $F_i$.
The principal does not know $v_i$ but knows the biases, which are not random.
We consider a generalization with random biases in Section~\ref{sec:applications-delegation-random-bias}.
The principal must select a nonempty subset $S\subseteq [n]$.
After the principal selects $S$, the agent observes $v_1,\ldots,v_n$ and selects the action $i^*$ maximizing $v_i+b_i$.
The principal's utility is $v_{i^*}$.

This problem reduces to utility configuration as follows: for each action $i \in [n]$, create a corresponding action in the utility configuration problem, with two configurations. \textit{in} and \textit{out}.
If action $i$ is \textit{in}, then $(u_i^A,u_i^P)$ are drawn by drawing $v_i$ from $F_i$ and setting  $(u_i^A,u_i^P)=(v_i+b_i,v_i)$.
If $i$ is \textit{out}, then $(u_i^A,u_i^P)=(-\infty,0)$ deterministically.
There is a natural bijection between configurations and solutions to delegated choice: the principal includes $i$ if and only if it is \textit{in}.
We next prove that delegated choice is $2$-utility aligned.

\begin{lemma}\label{lem:utility-alignment-delegation}
    Delegated choice is $2$-utility aligned.
\end{lemma}
\begin{proof}
    \newcommand{\vopt}{v_{i^*(\opt)}}
    \newcommand{\bopt}{b_{i^*(\opt)}}
    \newcommand{\istaropt}{i^*(\opt)}
    Given an arbitrary delegation set $S$, let $i^*(S)$ denote the random variable representing the agent's favorite action from $S$, given the realization of $(v_1,\ldots, v_n)$.
    Let $\opt$ denote the optimal delegation set, and consider any utility level $U$ for the agent. 
    Let $X=\E[v_{i^*(\opt)}~|~v_{i^*(\opt)}+b_{i^*(\opt)}\leq U]$ denote the principal's expected utility from the event that the agent's utility is below $U$.
    We will show that there exists a delegation set that gives the principal expected utility at least $X/2$, and hence $\E[v_{i^*(\opt)}]\geq X/2$ as well, as $\opt$ is optimal.
    To do so, we partition $\opt$ into {\em good} actions $G=\{i~|~b_i\leq U-X/2\}$ and {\em bad} actions $B=\{i~|~b_i> U-X/2\}$.
    We claim that $\E[v_{i^*(G)}]\geq X/2$.

    Let $\mathcal E$ denote the event that $v_{i^*(\opt)}+b_{i^*(\opt)}\leq U$, and $\mathcal E_G,\mathcal E_B$ the events that the agent's choice from $\opt$ lies in $G$ or $B$, respectively.
    First, note that conditioned on $\mathcal E$, it must be that $v_i\leq X/2$ for all $i\in B$:
    \begin{equation*}
        \val_i \leq \utilityThreshold - \bias_i < \utilityThreshold - \utilityThreshold + \expectedUtilityBelow/2 = \expectedUtilityBelow/2.
    \end{equation*}
    One consequence of this fact is that good actions represent at least half of $X$.
    That is:
    \begin{align*}
        X&=\E[v_{i^*(\opt)}~|~\mathcal E]\\
        &=\E[v_{i^*(\opt)}~|~\mathcal E\cap\mathcal E_G]\Pr[\mathcal E_G~|~\mathcal E]+\E[v_{i^*(\opt)}~|~\mathcal E\cap\mathcal E_B]\Pr[\mathcal E_B~|~\mathcal E]\\
        &\leq\E[v_{i^*(\opt)}~|~\mathcal E\cap\mathcal E_G]\Pr[\mathcal E_G~|~\mathcal E]+\frac{X}{2},
    \end{align*}
    and hence $\E[v_{i^*(\opt)}~|~\mathcal E\cap\mathcal E_G]\Pr[\mathcal E_G~|~\mathcal E]\geq X/2$.

    Meanwhile, let $\mathcal E_G^{U}$ denote the event that $\val_{\agentPreferredOption(G)} + \bias_{\agentPreferredOption(G)} \leq \utilityThreshold$.
    Then conditioned on $\overline{\mathcal E_G^U}$, note that every good action yields utility at least $X/2$:
    \begin{equation*}
        \val_i > \utilityThreshold - \bias_i \geq \utilityThreshold - \utilityThreshold + \expectedUtilityBelow/2 = \expectedUtilityBelow/2.
    \end{equation*}
    We can now show that $\E[v_{i^*(G)}]\geq X/2$. 
    We have:
    \begin{align}
        \E[\val_{\agentPreferredOption(G)}] &= \Pr[\mathcal E_G^{U}] \cdot \E[\val_{\agentPreferredOption(G)} | \mathcal E_G^U] + (1-\Pr[\mathcal E_G^{U}]) \cdot \E[\val_{\agentPreferredOption(G)} | \overline{\mathcal E_G^{U}}]\notag\\
        &\geq  \Pr[\mathcal E_G^{U}] \cdot \E[\val_{\agentPreferredOption(G)} | \mathcal E_G^U] + (1-\Pr[\mathcal E_G^{U}]) \cdot X/2\notag\\
        &=  \Pr[\mathcal E_G^{U}] \cdot \E[\val_{\agentPreferredOption(G)} | \mathcal E] + (1-\Pr[\mathcal E_G^{U}]) \cdot X/2\label{eq:del1}\\
        &=  \Pr[\mathcal E_G^{U}] \cdot \E[\val_{\agentPreferredOption(G)} | \mathcal E\cap \mathcal E_G]\prob[\mathcal E_G~|~\mathcal E] + (1-\Pr[\mathcal E_G^{U}]) \cdot X/2\label{eq:del2}\\
        &=  \Pr[\mathcal E_G^{U}] \cdot \E[\val_{\agentPreferredOption(\opt)} | \mathcal E\cap \mathcal E_G]\prob[\mathcal E_G~|~\mathcal E] + (1-\Pr[\mathcal E_G^{U}]) \cdot X/2\label{eq:del3}\\
        &\geq \Pr[\mathcal E_G^{U}] \cdot X/2 + (1-\Pr[\mathcal E_G^{U}]) \cdot X/2\notag\\
        &=X/2.\notag
    \end{align}
    Line (\ref{eq:del1}) follows from the independence of good and bad actions' values.
    Line (\ref{eq:del2}) follows from the fact that values are nonnegative.
    Line (\ref{eq:del3}) then follows from the fact that $\mathcal E_G$ implies that $\agentPreferredOption(\opt)=\agentPreferredOption(G)$.
    The remaining lines follow from the facts above.
\end{proof}

We also show that the analysis above is tight. 
\begin{lemma}
     Delegated choice is not $c$-utility aligned for any $c < 2$.
\end{lemma}
\begin{proof}
     Consider the following instance with 2 actions. The biases for these actions are $\bias_1 = \bucketsCount-1 + 1/\bucketsCount$ and $\bias_2 = 0$.
     \[
         \val_1 = \begin{cases}
             1, &w.p.~ 1-1/\bucketsCount\\
             1-2/\bucketsCount, &w.p.~ 1/\bucketsCount
         \end{cases} \hspace{2cm}
         \val_2 = \begin{cases}
             \bucketsCount, &w.p.~ 1/\bucketsCount\\
             0, &w.p.~ 1-1/\bucketsCount
         \end{cases}
     \]

     In this case, it can be easily checked that the optimal delegation set includes both actions, which results in the following principal utility distribution:
     \[
         \val_{\agentPreferredOption(\opt)} = \begin{cases}
             1, &w.p.~ 1-1/\bucketsCount\\
             \bucketsCount, &w.p.~ (1/\bucketsCount)^2\\
             1-2/\bucketsCount, &w.p.~ (1-1/\bucketsCount)/\bucketsCount
         \end{cases}
     \]
    
    The total principal utility is equal to  \[(1-1/\bucketsCount)+(1/\bucketsCount)
    + (1-1/\bucketsCount)(1-2/\bucketsCount)/\bucketsCount =1+(1-1/\bucketsCount)(1-2/\bucketsCount)/\bucketsCount\] 
    
    Notice that the agent's utility from the $1/\bucketsCount$-quantile comes from the last two realizations of $\val_{\agentPreferredOption(\opt)}$. The principal's utility conditioned on the last two realizations is 
    \[
    (1/\bucketsCount) \cdot \bucketsCount + (1-1/\bucketsCount) \cdot (1-2/\bucketsCount)=1+(1-1/\bucketsCount)(1-2/\bucketsCount).
    \]
    
    If delegated choice were $c$-utility aligned for some $c<2$, it would have to be that
    \[
    \frac{1+(1-1/\bucketsCount)(1-2/\bucketsCount)}{1+(1-1/\bucketsCount)(1-2/\bucketsCount)/\bucketsCount} \leq c,
    \]
    which is a contradiction, since the limit of this expression (as $M\to\infty$) is $2$.
\end{proof}

Applying the analysis of the PTAS from the previous section--- in particular, Corollary~\ref{cor:2}--- we immediately obtain the following:
\begin{corollary}
    For any fixed $M$, the algorithm from Section~\ref{sec:main-results} with $M$ bins is a $(1-O(\log M/M))$-approximation for the delegated choice problem and runs in polynomial time.
\end{corollary}

%% file: applications/pricing.tex
We next consider the problem of pricing $n$ items for sale to a unit-demand agent.
Recall that for this problem, we have $n$ items.
Item $i$ has value $v_i$ for the buyer, which has distribution $F_i$.
The buyer's values for the different items are independent.
The principal (seller) selects a price $p_i\geq 0$ for each item $i$.
The buyer then realizes their values and picks the single item $i^*$ that maximizes their profit $v_{i^*}-p_{i^*}$, or leaves without buying an item if $v_i-p_i<0$ for all $i$.
The seller's utility is their revenue, equal to $p_{i^*}$ (or $0$ if no item is sold).

Unit-demand pricing reduces to utility configuration: each item will represent an action in the configuration instance.
When the seller can use any positive real-valued price, the resulting utility configuration problem will have a continuum of configurations per action.
While this is a mathematically well-defined optimization problem, the PTAS of the previous section requires there to be a finite number of configurations per action.
This can be solved via standard price discretization techniques, e.g.\ those from \citet{CD11}.
We discuss the reduction from continuous to discrete pricing in Section~\ref{sec:pricing}, where we also show how to obtain a PTAS for unit-demand pricing with unbounded regular value distributions.
For the time being, we prove that unit-demand pricing is $2$-utility aligned for any set of prices available to the seller, as long as this set includes $\infty$ (i.e., not selling an item).

The reduction from unit-demand pricing to utility configuration is the following.
Each item $i$ corresponds to an action in the configuration problem.
Each price $p_i$ for item $i$ corresponds to a configuration for action $i$.
To draw the utilities for item $i$ priced at $p_i$, draw $v_i$ from $F_i$.
If $v_i\geq p_i$, then $(u_i^A,u_i^P)=(v_i-p_i,p_i)$.
Otherwise, $=(u_i^A,u_i^P)=(0,0)$.
Using a very similar argument to the one for delegation (Lemma~\ref{lem:utility-alignment-delegation}), we can prove the following:
\begin{lemma}
    Unit-demand pricing is $2$-utility aligned as long as the seller's set of allowable prices contains $\infty$.
\end{lemma}
\begin{proof}
    Given price vector $\prices$, let $i^*(\prices)$ denote the buyer's utility-maximizing choice of item, or $0$ if no item yields nonnegative utility.
    Then the seller's revenue is $\E[p_{i^*(\prices)}]$, where we adopt the convention that $p_0=0$ and $v_0=0$.
    Now let ${\prices}$ denote the optimal price vector, and consider a utility level $U$ for the buyer.
    Let $\mathcal E$ denote the event that $v_{i^*({\prices})}-p_{i^*({\prices})}\leq U$.
    Let $X=\E[\optrev~|~\mathcal E]$.
    We will show that there is a pricing $\prices'$ where the seller earns at least $X/2$ in expectation.
    To pick $\prices'$, let $p_i'= p_i$ for all $i$ such that $ p_i\geq X/2$, and $p_i'=\infty$ for all other items.
    Call the former set of items the {\em good} items $G$ and the latter set the {\em bad} items $B$.

    The claim will follow from two key observations.
    First, we observe that the good items represent at least half the seller's revenue conditioned on $\mathcal E$.
    To see this, let $\mathcal E_G$ denote the event that the agent buys a good item given $\prices$, and $\mathcal E_B$ similarly for a bad item.
    We have:
    \begin{align*}
        X&=\E[p_{i^*( \prices)}~|~\mathcal E]\\
        &=\E[p_{i^*( \prices)}~|~\mathcal E\cap\mathcal E_G]\Pr[\mathcal E_G~|~\mathcal E]+\E[p_{i^*( \prices)}~|~\mathcal E\cap\mathcal E_B]\Pr[\mathcal E_B~|~\mathcal E]\\
        &\leq\E[p_{i^*( \prices)}~|~\mathcal E\cap\mathcal E_G]\Pr[\mathcal E_G~|~\mathcal E]+\frac{X}{2},
    \end{align*}
    and hence $\E[p_{i^*( \prices)}~|~\mathcal E\cap\mathcal E_G]\Pr[\mathcal E_G~|~\mathcal E]\geq X/2$.

Now let $\mathcal E_G^U$ denote the event that the buyer's favorite item under $\prices'$ yields utility less than $U$, i.e. $v_{i^*(\prices')}-p_{i^*(\prices')}\leq U$.
The second observation is that conditioned on $\mathcal E_G^U$, the buyer is more likely to choose any $i\in G$ given $\prices'$ than they are to choose $i$ given $\prices$ conditioned on $\mathcal E$.
Intuitively, in the former scenario, item $i$ has the same distribution but less competition. Formally, we have the following for all $i\in G$:
\begin{align}
    \Pr[v_i-p_i'\geq v_{i'}-p_{i'}'~\forall i'\in G~|~\mathcal E_G^{U}]&=\Pr[v_i-p_i\geq v_{i'}-p_{i'}~\forall i'\in G~|~\mathcal E_G^{U}]\label{eq:priceprob1}\\
    &=\Pr[v_i-p_i\geq v_{i'}-p_{i'}~\forall i'\in G~|~\mathcal E]\label{eq:priceprob2}\\
    &\geq \Pr[v_i-p_i\geq v_{i'}-p_{i'}~\forall i'\in[n]~|~\mathcal E]\label{eq:priceprob3}.
\end{align}
Line (\ref{eq:priceprob1}) holds as $p_i=p_i'$ for all $i\in G$.
Line (\ref{eq:priceprob2}) follows from the independence of values between items.
Line (\ref{eq:priceprob3}) follows because we have added more constraints compared to the line before.

The following inequalities then prove the lemma:
\begin{align}
        \E[p_{\agentPreferredOption(\prices')}'] &= \Pr[\mathcal E_G^{U}] \cdot \E[p_{\agentPreferredOption(\prices')}' | \mathcal E_G^U] + (1-\Pr[\mathcal E_G^{U}]) \cdot \E[p_{\agentPreferredOption(\prices')}' | \overline{\mathcal E_G^{U}}] \notag\\
        &\geq  \Pr[\mathcal E_G^{U}] \cdot \E[p_{\agentPreferredOption(\prices')}' | \mathcal E_G^U] + (1-\Pr[\mathcal E_G^{U}]) \cdot X/2\label{eq:udp2}\\
        &=  \Pr[\mathcal E_G^{U}] \sum_{i\in G}p_i'\Pr[v_i-p_i'\geq v_{i'}-p_{i'}'~\forall i'\in G~|~\mathcal E_G^{U}] + (1-\Pr[\mathcal E_G^{U}]) \cdot X/2\label{eq:udp4}\\
        &\geq  \Pr[\mathcal E_G^{U}] \sum_{i\in G}p_i\Pr[v_i-p_i\geq v_{i'}-p_{i'}~\forall i'\in[n]~|~\mathcal E] + (1-\Pr[\mathcal E_G^{U}]) \cdot X/2\label{eq:udp7}\\
         &=  \Pr[\mathcal E_G^{U}]\E[p_{i^*(\prices)}~|~\mathcal E\cap\mathcal E_G]\Pr[\mathcal E_G~|~\mathcal E] + (1-\Pr[\mathcal E_G^{U}]) \cdot X/2\notag\\
        &\geq \Pr[\mathcal E_G^{U}] \cdot X/2 + (1-\Pr[\mathcal E_G^{U}]) \cdot X/2\notag\\
        &=X/2.\notag
    \end{align}
    Line (\ref{eq:udp2}) holds because whenever there is some good item $i$ with $v_{i^*(\prices')}-p_{i^*(\prices')}\geq U\geq 0$, the buyer purchases a good item under $\prices'$.
    Line (\ref{eq:udp4}) rewrites the seller's revenue under $\mathcal E_G^U$ in terms of the purchase probabilities for the individual items.
    Line (\ref{eq:udp7}) holds because item $i$ is only more likely to be bought if the bad actions are priced at $\infty$ than not, and because $p_i=p_i'$ for all $i\in G$.
    Since $\prices'$ yields at least $X/2$, the optimal pricing must yield at least this much revenue as well, proving the lemma.
     \end{proof}

    While the framework of Section~\ref{sec:main-results} is not compatible with infinitely many prices (and hence configurations of each action), for any fixed set of $m$ allowable prices per item (including $\infty$), we obtain the following by applying Corollary~\ref{cor:2}.
\begin{corollary}
    For any fixed $M$, the algorithm from Section~\ref{sec:main-results} with $M$ bins is a $(1-O(\log M/M))$-approximation for the unit-demand pricing problem with $m$ possible prices, and runs in polynomial time.
\end{corollary}

We show in Section~\ref{sec:pricing} how to reduce the infinite-prices problem to the version with finitely many prices.
For discrete value distributions, the number of prices required will be polynomial in the bit complexity of the values, leading to a PTAS for that problem as well.
For unbounded regular distributions, the number of prices required will be polynomial in the number of items.

%% file: applications/general/general.tex
Sections~\ref{sec:applications-delegation} and \ref{sec:applications-pricing} showed that both delegated choice and unit-demand pricing were $2$-utility aligned.
Both proofs relied heavily on the structures of the problems: both involved imposing a threshold on an existing solution (either based on bias or price) and discarding all items or actions below the threshold.
For general instances of utility configuration, such an approach may not be possible.
However, in this section, we show a weaker utility alignment result that applies to a wide range of natural problems.
The guarantee will still be strong enough to obtain a PTAS using the algorithm of Section~\ref{sec:main-results}.

Definition \ref{def:alignment} of utility alignment is a statement about optimal solutions: an instance is utility aligned if, under the optimal solution, the principal's and agent's utilities are aligned.
The result of this section considers an item-wise version of this same condition, defined below.
\begin{definition}\label{def:local}
    For non-increasing function $f$, an instance of utility configuration is {\em locally $f(q)$-utility aligned} if for every action $i$, configuration $j$ and utility level $U$ such that $\Pr[u_i^A\leq U]= q$, we have:
\[
        \E_{(\utilAgent{i},\utilPrincipal{i})\sim F_{ij}}[\utilPrincipal{i} | \utilAgent{i} \leq \utilityThreshold] \leq f(\q) \cdot \E_{(\utilAgent{i},\utilPrincipal{i})\sim F_{ij}}[\utilPrincipal{i}].
    \]
\end{definition}
Note that Definition~\ref{def:local} is satisfied by both delegated choice (Section~\ref{sec:applications-delegation}) and unit-demand pricing (Section~\ref{sec:applications-pricing}), each with $f(q)$ identically equal to $1$ (i.e.\ both problems are {\em locally $1$-utility aligned}).
For example, in delegated choice, the agent's and principal's utilities are perfectly and positively correlated: for actions which are included in the delegation set, $u_i^P=v_i$ increases only if $u_i^A=v_i+b_i$ increases.
The following result shows that local utility alignment (Definition~\ref{def:local}) is sufficient for utility alignment overall (Definition \ref{def:alignment}), albeit with a weaker dependence of $f$ on $q$ than the results for delegated choice and pricing.
The result also requires one more natural property --- which again holds for all problems in this paper --- and which states that the principal can choose to completely discard any action.
For example, in delegation, the principal's main decision is exactly this choice of including or excluding actions.
\begin{definition}\label{def:excludable}
    A utility configuration instance has {\em excludable actions} if for every action $i$, there is a configuration such that $u_i^A=-\infty$ deterministically in that configuration.
\end{definition}

We say action $i$ is \emph{excluded} in a configuration $\config$ if $C_i$ is such that $\utilAgent{i} = -\infty$ deterministically. 
The result is then the following:

\begin{lemma}\label{lem:localtoglobal}
    Any instance of utility configuration with excludable actions that satisfies local $f(q)$-utility alignment also satisfies $\max(4f(q),4/\sqrt{q})$-utility alignment.
\end{lemma}

As a consequence of Lemma~\ref{lem:localtoglobal}, local $O(1/\sqrt{q})$-utility alignment also translates into overall $O(1/\sqrt{q})$-utility alignment.
We will show in Sections~\ref{sec:applications-delegation-random-bias}-\ref{sec:applications-assortment} that two generalizations of the delegated choice problem, as well as the assortment optimization problem from the revenue management literature, all in fact satisfy local $1$-utility alignment, and hence are $O(1/\sqrt{q})$-aligned overall.

\begin{proof}[Proof of Lemma~\ref{lem:localtoglobal}]
    Let $\opt$ denote the optimal configuration, and given any configuration $\config$, let $\agentPreferredOption(\config)$ be the random variable denoting the agent's choice under $\config$.
    Given utility threshold $U$, let $\mathcal E$ denote the event that $\utilAgent{\agentPreferredOption(\OPT)} \leq \utilityThreshold$, and let $q=\Pr[\mathcal E]$.
    Further let $X=\E[\utilPrincipal{\agentPreferredOption(\OPT)} | \utilAgent{\agentPreferredOption(\OPT)} \leq \utilityThreshold]$.
    The proof considers two cases.
    The main goal is to break the actions into two groups with roughly equal contribution to $X$, and to show that discarding the worse of these two groups guarantees the principal utility at least $\sqrt{q}$.
    The only case where this strategy fails is when one action contributes a large fraction of $X$.
    In this case, local utility alignment implies that this action alone will provide high utility.

    \paragraph{Case 1.} Assume that for every action $i$, the contribution to $X$ given by $\E[\utilPrincipal{\agentPreferredOption(\OPT)} | \mathcal E\cap \agentPreferredOption(\OPT)=i]\Pr[\agentPreferredOption(\OPT)=i~|~\mathcal E]$ is at most $X/4$.
    Then we may partition the actions into two sets $S_1$ and $S_2$ such that for $j\in \{1,2\}$,
    \begin{equation*}
        \E[\utilPrincipal{\agentPreferredOption(\OPT)} | \mathcal E\cap \mathcal E_j]\Pr[\mathcal E_j~|~\mathcal E]\geq X/4,
    \end{equation*}
    where $\mathcal E_j$ denotes the event that $\agentPreferredOption(\OPT)\in S_j$.
    Let $\config^j$ denote the configuration where actions not in $S_j$ are excluded, and those in $S_j$ are configured as in $\opt$.
    Let $\agentPreferredOption(S_j,\opt)$ denote the agent's favorite action from $S_j$ under $\opt$, and recall that $\agentPreferredOption(\config^j)$ is the agent's favorite action under $\config^j$.
    Note that $\agentPreferredOption(S_j,\opt)$ and $\agentPreferredOption(\config^j)$ have the same distribution.
    Let $\mathcal E_j^U$ denote the event that the agent's favorite action under $\config^j$ gives them utility at most $U$, or equivalently that $u_{\agentPreferredOption(S_j)}^A\leq U$, and let $\epsilon_j=\Pr[\mathcal E_j^U]$.
    
    The result then follows from two observations.
    First, $\config^j$ guarantees the principal at least utility $X/4\cdot \epsilon_{j}$.
    In particular, we may express the principal's utility from this configuration as:
    \begin{align}
        \E[\utilPrincipal{\agentPreferredOption(\config^j)}]&\geq\E[\utilPrincipal{\agentPreferredOption(\config^j)}~|~\mathcal E_j^U]\cdot\epsilon_j\label{eq:genalign1}\\
        &=\E[\utilPrincipal{\agentPreferredOption(S_j,\opt)}~|~\mathcal E]\cdot\epsilon_j\label{eq:genalign2}\\
        &\geq\E[\utilPrincipal{\agentPreferredOption(S_j,\opt)}~|~\mathcal E\cap \mathcal E_j]\Pr[\mathcal E_j~|~\mathcal E]\cdot\epsilon_j\label{eq:genalign3}\\
        &\geq\E[\utilPrincipal{\agentPreferredOption(\opt)}~|~\mathcal E\cap \mathcal E_j]\Pr[\mathcal E_j~|~\mathcal E]\cdot\epsilon_j\label{eq:genalign4}\\
        &\geq X/4\cdot\epsilon_j\label{eq:genalign5}
    \end{align}
    Line (\ref{eq:genalign1}) follows from the fact that the principal utility is nonnegative and the definition of $\epsilon_j$.
    Line (\ref{eq:genalign2}) follows from the independence of utilities across actions.
    Line (\ref{eq:genalign3}) then again follows from the non-negativity of principal utilities.
    Line (\ref{eq:genalign4}) follows from the fact that in event $\mathcal E_j$, the agent chooses an action from $S_j$, so $\agentPreferredOption(\opt)=\agentPreferredOption(S_j,\opt)$.
    Finally, Line (\ref{eq:genalign5}) holds because $\E[\utilPrincipal{\agentPreferredOption(\OPT)} | \mathcal E\cap \mathcal E_j]\Pr[\mathcal E_j~|~\mathcal E]\geq X/4$.

    To conclude the proof of utility alignment in this case, notice that since $S_1$ and $S_2$ partition the actions, $q=\epsilon_1\cdot\epsilon_2$.
    Hence, $\max_{j\in\{1,2\}} \epsilon_j\geq \sqrt{q}$.
    Hence, choosing the $j$ with the larger $\epsilon_j$ yields a set with expected principal utility at least $X\sqrt{q}/4$, as desired.

    \paragraph{Case 2.} Assume some action $i$ contributes at least a quarter of $X$,
    i.e.\ $\E[\utilPrincipal{\agentPreferredOption(\OPT)} | \mathcal E\cap \agentPreferredOption(\OPT)=i]\Pr[\agentPreferredOption(\OPT)=i~|~\mathcal E]\geq X/4$.
    Consider the configuration $\hat\config^i$ where $i$ is configured as in $\opt$, but all other actions are excluded.
    Then $\hat\config^i$ guarantees the principal utility at least $X/(4f(q))$.
    To see this, let $\mathcal E_i^U$ denote the event that $u_i^A\leq U$, and $\mathcal E_i$ the event that $\agentPreferredOption(\OPT)=i$.
    Then we have the following inequalities, justified after their statement:
    \begin{align}
    \E[u_i^P]&\geq \frac{1}{f(\Pr[u_i^A\leq U])}\E[u_i^P~|~\mathcal E_i^U]\label{eq:single1}\\
    &\geq \frac{1}{f(q)}\E[u_i^P~|~\mathcal E_i^U]\label{eq:single2}\\
    &= \frac{1}{f(q)}\E[u_i^P~|~\mathcal E]\label{eq:single3}\\
    &\geq \frac{1}{f(q)}\E[u_i^P~|~\mathcal E\cap \agentPreferredOption(\OPT)=i]\Pr[\agentPreferredOption(\OPT)=i~|~\mathcal E]\label{eq:single4}\\
    &\geq \frac{X}{4f(q)}.\label{eq:single5}
    \end{align}
    Line (\ref{eq:single1}) follows from local utility alignment.
    Line (\ref{eq:single2}) holds because $f$ is non-increasing, and because $\Pr[u_i^A\leq U]\geq \Pr[u_{\agentPreferredOption(\opt)}^A\leq U]=q$.
    Line (\ref{eq:single3}) holds because of the independence between actions.
    Line (\ref{eq:single4}) holds because principal utilities are nonnegative.
    Finally, Line (\ref{eq:single5}) holds because $i$ contributes at least $1/4$ of $X$.
    The overall utility alignment is the worse of the two guarantees from the two cases.
    This leads to the stated bound.
\end{proof}

In the rest of this section, we prove that several further applications of interest can be reduced to utility configuration, and furthermore that these applications are all locally $1$-utility aligned.
Consequently, all the applications below are also $4/\sqrt{q}$-utility aligned overall and hence the algorithm of Section~\ref{sec:main-results} is a PTAS.
In Sections~\ref{sec:applications-delegation-random-bias} and \ref{sec:applications-delegation-outside-option}, we consider two generalizations of the delegated choice problem.
In Section~\ref{sec:applications-delegation-random-bias}, we allow each action $i$ to have bias $b_i$, which is random, but independent from the action's value $v_i$, as well as the biases and values of other actions.
Section~\ref{sec:applications-delegation-outside-option} then considers a variant of delegation where the agent has an outside option, an action which gives the principal no utility and cannot be excluded from the delegation set.
Finally, Section~\ref{sec:applications-assortment} considers the well-studied assortment optimization problem from revenue management.

%% file: applications/general/delegation_random.tex
In the delegated choice problem considered in Section~\ref{sec:applications-delegation}, each action's value $v_i$ was randomly distributed, and the action's bias $b_i$ was fixed and known to the principal.
We now allow the bias $b_i$ to be random as well, with the distribution also known to the principal.
This can capture scenarios where the bias stems from, e.g.\ measurement error, as well as from uncertainty of the principal regarding the agent's preferences.
For example, one could consider $b_i\sim N(0,\sigma_i^2)$, with variance that varies across actions.
The agent then myopically chooses the action that maximizes the measured value.

This problem reduces to utility configuration in exactly the same way that the fixed-bias version does.
Each action has an {\em in} configuration and an {\em out} configuration.
The {\em out} configuration gives agent $-\infty$ and principal $0$ utility deterministically, as in the vanilla delegation model.
The distribution of $(u_i^A,u_i^P)$ in the {\em in} configuration is equal to that of $(v_i+b_i,v_i)$, drawing $v_i$ and $b_i$ independently from their respective distributions.
We then have the following local utility alignment guarantee:

\begin{lemma}\label{lem:randombias}
    Delegation with randomly distributed bias is locally $1$-utility aligned.
\end{lemma}
\begin{proof}
    We only need to consider the {\em in} configuration for any action $i$.
    Fixing a utility threshold $U$, we have:
    \begin{align*}
    \E[v_i~|~v_i+b_i\leq U]&=\E_{b_i}\left[\E_{v_i}[v_i~|~v_i\leq U-b_i]~|~ b_i\leq U\right]\\
        &\leq\E_{b_i}\left[\E_{v_i}[v_i]~|~ b_i\leq U\right]\\
        &=\E[v_i].
    \end{align*}
    This immediately implies the desired local utility alignment result.
\end{proof}
Note that Lemma~\ref{lem:randombias} will also tend to hold when $v_i$ and $b_i$ satisfy natural positive correlation conditions, via essentially the same proof.
Combining Lemma~\ref{lem:randombias} with the local-to-global result of Lemma~\ref{lem:localtoglobal} and the guarantee of Corollary~\ref{cor:sqrt}, we obtain:

\begin{corollary}
    For any fixed $M$, the algorithm from Section~\ref{sec:main-results} with $M$ bins is a $(1-O(1/\sqrt{M}))$-approximation for the delegated choice problem with random biases and runs in polynomial time.
\end{corollary}

%% file: applications/general/delegation_outside_option.tex
The second generalization to delegated choice we consider is delegation with an {\em outside option}.
We assume the agent has an action that they may take, which yields utility $v_0=0$ for the principal, but utility $b_0$ for the agent.
The agent then chooses the action maximizing $v_i+b_i$ if this quantity exceeds $b_0$, or takes the outside option otherwise, yielding $0$ utility for the principal.
As in Section~\ref{sec:applications-delegation-random-bias}, we allow biases, including $b_0$, to be random and independent of each other and the values.

The obvious approach to reduce this problem to utility alignment fails: if we follow the reduction of Section~\ref{sec:applications-delegation-random-bias}, then add a $0$-th action corresponding to the outside option, this $0$-th action needs to have one configuration (the ``{\em in}'' configuration), as the principal cannot exclude the outside option.
However, this construction now fails to satisfy the excludable actions assumption (Definition~\ref{def:excludable}).

Instead, we adjust the utility distributions for non-outside actions in the following way.
As before, we let the {\em out} configuration for each action yield utility $-\infty$ for the agent and $0$ for the principal.
Now for action $i$ in the {\em in} configuration, draw $(u_i^A,u_i^P)$ in the following way:
first draw $v_i$ and $b_i$ from their respective distributions.
Set $u_i^A=v_i+b_i$, as before.
Let $u_i^P$ be set as $ u_i^P=v_i\cdot\Pr_{b_0}[v_i+b_i\geq b_0]$.
Given a configuration of actions in this utility configuration instance, it is straightforward to see that the principal's expected utility is the same as if they took the corresponding actions in the delegation instance.
Intuitively, the utilities above simulate the following two-step process: first, the agent chooses their favorite action other than the outside option, then switches to the outside option if $b_0$ is high enough.
Intuitively, this preserves local utility alignment: if the agent's utility is higher, then they are less likely to switch to the outside option and deny utility to the principal.

\begin{lemma}\label{lem:outside}
    Delegation with an outside option and randomly distributed bias is locally $1$-utility aligned.
\end{lemma}
\begin{proof}
The proof generalizes that of Lemma~\ref{lem:randombias}. Let $a(v_i,b_i)=v_i\prob_{b_0}[v_i+b_i\geq b_0]$ denote the contribution of action $i$ after factoring in the outside option. Note that $a(v_i,b_i)$ is nondecreasing in both $v_i$ and $b_i$. 
Fixing a utility threshold $U$, let $G^U_i$ be the conditional distribution of $b_i$ where $b_i\leq U$ and let
$F^{(U-b_i)}$ be the conditional distribution of $v_i$ where $v_i\leq U-b_i$. Then can write
    \begin{align*}
    \E[a(v_i,b_i)~|~v_i+b_i\leq U]&=\E_{b_i\sim G_i^U}\left[\E_{v_i\sim F_i^{(U-b_i)}}[a(v_i,b_i)]\right]\\ &\leq 
    \E_{b_i\sim G_i^U}\left[\E_{v_i\sim F_i}[a(v_i,b_i)]\right]\\
     &\leq 
    \E_{b_i\sim G_i}\left[\E_{v_i\sim F_i}[a(v_i,b_i)]\right]\\
     &=     \E[a(v_i,b_i)],
    \end{align*}
where the inequalities follow from the fact $G$ and $F$ respectively first-order stochastically dominate $G^U$ and $F^{(U-b)}$, and because $a(v_i,b_i)$ is nondecreasing in $v_i$ and $b_i$.
This implies local $1$-utility alignment.
\end{proof}

Lemmas~\ref{lem:outside} and Lemma~\ref{lem:localtoglobal} combine with Corollary~\ref{cor:sqrt} to imply:

\begin{corollary}
    For any fixed $M$, the algorithm from Section~\ref{sec:main-results} with $M$ bins is a $(1-O(1/\sqrt{M}))$-approximation for the delegated choice problem with an outside option, and runs in polynomial time.
\end{corollary}

%% file: applications/general/assortment_optimization.tex
The final application we consider is assortment optimization, with a random utility model and independent utilities across items.
In the assortment optimization problem, there are $n$ fixed-price items, where item $i$ has price $p_i$.
The seller's decision is to pick an inventory, a subset of $\{1,\ldots,n\}$ to offer for sale to a single buyer.
The buyer's value for item $i$ is $v_i$, drawn independently from distribution $F_i$, yielding utility $v_i-p_i$ for buying item $i$.
The buyer also has an outside option, which yields utility $u_0\sim F_0$ for the buyer and $0$ for the seller.
The seller seeks to maximize their revenue in expectation over the random choice of the agent.

This problem reduces to delegation with an outside option (and, by composition, to utility configuration) in the following way.
We define delegation value $v_i' = p_i$ and (random) delegation bias $b_i'$ distributed as $v_i-2p_i$.
Note that since $p_i$ is fixed, $v_i'$ is fixed and hence $v_i'$ and $b_i'$ are independent.
Include in the delegation instance an outside option with agent utility distributed as $F_0$.
The resulting instance is equivalent to the original assortment optimization problem.
Through the reduction and utility alignment results of Section~\ref{sec:applications-delegation-outside-option}, then, we obtain the following:

\begin{lemma}
    Assortment optimization with independent utilities is locally $1$-utility aligned.
\end{lemma}

\begin{corollary}
    For any fixed $M$, the algorithm from Section~\ref{sec:main-results} with $M$ bins is a $(1-O(1/\sqrt{M}))$-approximation for assortment optimization, and runs in polynomial time.
 \end{corollary}

%% file: pricing/pricing.tex
The algorithm for utility configuration in Section~\ref{sec:main-results} assumed a finite number $m$ of configurations per action.
Its runtime was polynomial in this number, among the other parameters of the input size.
However, the usual statement of the unit-demand pricing problem allows the seller to offer arbitrary nonnegative prices, and the reduction from unit-demand pricing to utility configuration translates every price for each item $i$ into a configuration for $i$'s corresponding action, leading to an instance with a continuum of configurations.
Since this is clearly not computationally tractable, in this section we re-purpose lemmas from prior work on the problem \citep{CD11} to show that a number of prices per item --- and hence configurations per action --- that is polynomial in the input size suffices for a $(1-\epsilon)$-approximation, and can be computed in polynomial time.
Combined with our results, this implies a PTAS for the case with discrete values but unrestricted prices.
This yields the first multiplicative PTAS for this case, improving on the additive PTAS of \citet{CD11}.
These results are in Section~\ref{sec:pricing-finite-dist}.

We also study the case where values are continuously distributed but satisfy the standard Myerson regularity condition.
In this case, \citet{CD11} give a quasi-polynomial time approximation scheme for the regular case, based on a novel lemma for truncating the value distributions, and another lemma for discretizing it into polynomially many mass points.
We combine these with our framework to produce the first PTAS for this problem.
These results can be found in Section~\ref{sec:pricing-regular-dist}.

\subsection{Finite Value Distributions}\label{sec:pricing-finite-dist}
\input{pricing/finite_distributions}
\subsection{Unbounded Regular Distributions}\label{sec:pricing-regular-dist}
\input{pricing/regular_distributions}

%% file: pricing/finite_distributions.tex
In the discrete-values version of unit-demand pricing, each item $i$'s value $v_i$ is drawn independently from a finitely supported discrete distribution $F_i$.
Note that even with discrete distributions, it may be in the seller's interest to offer prices that are not equal to any mass for any distribution.
(See, e.g.\ \citet{CD11}, Appendix J for an example where this holds.)
Consequently, it is necessary to show that finitely many prices suffice for a $(1-\epsilon)$-approximation, and that these prices can be computed efficiently.
Fortunately, we may borrow the following directly from \citet{CD11}:

\begin{lemma}[\citet{CD11}, Lemma 25] \label{lemma:pricing-discretization}
    Suppose that the value distributions in an instance of the item pricing problem are independent and supported on $[\umin, \umax] \subset \R_+$. For any $\varepsilon \in (0, 1/2)$, consider the following finite set of prices:
    \[
        \priceSet_\varepsilon = \set{\price_i~|~\price_i = \frac{1+\varepsilon^2-\varepsilon}{(1-\varepsilon^2)^i} \cdot \umin, i \in [\lfloor \log_{\frac{1}{1-\varepsilon^2}} (\umax/\umin) \rfloor]}.
    \]
    For any price vector $P \in [\umin, \umax]^\optionCount$, there exists a price vector $P'$ such that $\price_i' \in \priceSet_\varepsilon$ and $\price_i' \in [1-\varepsilon, 1+\varepsilon^2-\varepsilon] \cdot \price_i$, for all $i$. The expected revenue achieved by the two price vectors satisfies $\rev_{P'} \geq (1-2\varepsilon)\rev_{P}$.
\end{lemma}

Given that prices below $u_{\text{min}}$ are clearly suboptimal, and above $u_{\text{max}}$ are dominated by $u_{\text{max}}$, Lemma~\ref{lemma:pricing-discretization} implies that we may obtain a multiplicative PTAS for the discrete-values case by first discretizing the prices, then applying the algorithm from Section~\ref{sec:main-results} to the version with discrete prices.
Hence, we have:
\begin{corollary}
    There exists a multiplicative PTAS for the unit-demand pricing problem with independent, finitely supported, discretely distributed values.
\end{corollary}

%% file: pricing/regular_distributions.tex
We now consider a version of the unit-demand pricing problem where values are continuously distributed and possibly unbounded.
Without some tail condition on the distribution, this problem is intractable.
A standard condition is {\em Myerson regularity}, which states that the virtual value function for each distribution $F_i$ with density $f_i$, given by $\phi_i(v_i)=v_i-(1-F_i(v_i))/f_i(v_i)$, is non-decreasing in $v_i$.
For regular distributions, \citet{CD11} give a QPTAS for unit-demand pricing, computing a $(1-\epsilon)$-approximation in time polynomial in $n^{\text{polylog}(n)}$.
In addition to the price discretization lemma in the previous section, they give a {\em truncation lemma} and a {\em value discretization lemma}.
The number of prices they produce in the end is dependent on the number of items, and the dependence of their algorithm's runtime on the number of prices is exponential.
Our algorithm has a polynomial dependence on the number of prices, which allows us to improve the result to a PTAS for the regular case.

We adopt the same model of distributional access as \cite{CD11}: their results (and hence our applications of these results as well) hold as long as the algorithm has access to one of the following: (1) explicit access to the CDFs, (2) oracle access to the CDFs, along with information about a fixed quantile of the distributions (e.g.\ the median); such an anchoring point is necessary, or else the algorithm would face the hopeless task of guessing the distribution's scale, (3) sample access to the distribution, in which case the guarantees hold with high probability over the samples.
For more details, see Appendix B of \citet{CD11}.

Given any of the three forms of access above, the first step is to truncate the potentially unbounded regular distribution.
The following lemma, again from \citet{CD11}, accomplishes this goal.
Given an instance with independent, unbounded regular distributions, \citet{CD11} show that it is possible to create a corresponding instance with distributions bounded in a polynomially-sized interval.
Given near-optimal prices for this new instance, it is possible to construct near-optimal prices for the original instance in polynomial time.
Formally:
\begin{lemma}[\citet{CD11}, Theorem 20]\label{lem:truncation}
Let $\mathcal V$ be a collection of mutually independent regular random variables.
Then there exists some $\alpha(\mathcal V)>0$ such that for any $\epsilon\in(0,1)$, there is a reduction from $(1-\epsilon)$-approximately optimal pricing for $\mathcal V$ to $(1-\theta(\epsilon/n))$-approximately optimal pricing for $\tilde{\mathcal V}$, where $\tilde{\mathcal V}$ is a collection of mutually independent random variables supported on $[\tfrac{\epsilon\alpha}{4n^4},\frac{4n^4\alpha}{\epsilon^3}]$.
Moreover, $\alpha$ is efficiently computable from the distributions of $v_i$s, and for all $\epsilon$, the runtime of the reduction is polynomial in the input size and $1/\epsilon$.
\end{lemma}

Finally, \citet{CD11} give a discretization lemma for bounded value distributions.
As with Lemma~\ref{lem:truncation} for truncating the distributions, they show that it is possible to convert continuous distributions into discrete distributions and to then convert approximately optimal prices for the discrete distributions into approximately optimal prices for the original instance.
The number of values in the support of the distributions is polylogarithmic in the ratio of the highest to the lowest value in the original distribution.
Formally:

\begin{lemma}[\citet{CD11}, Lemma 30]\label{lem:discretization}
    Let $\{v_i\}_{i\in[n]}$ be a product distribution over values supported on range $[u_{\text{min}},u_{\text{max}}]\subset \mathbb R_+$, and $r=\tfrac{u_{\text{max}}}{u_{\text{min}}}$.
    Then for any $\delta\in(0,(4\lceil\log_2 r\rceil)^{4/3})$, there exists a collection of mutually independent random variables $\{\hat v_i\}_{i\in[n]}$ which are supported on a set of $O(\log r/\delta^2)$ values that satisfy the following properties:
    \begin{enumerate}
        \item The optimal revenue under $\{\hat v_i\}_{i\in[n]}$ is at least a $(1-3\delta^{1/8})$-fraction of that under $\{v_i\}_{i\in[n]}$.
        \item For any $\rho\in(0,1/2)$ and any price vector $\prices$ that yield revenue at least a $(1-\rho)$-fraction of optimal for $\{\hat v_i\}_{i\in[n]}$, we can construct in polynomial time in the description of $\prices$ and $1/\delta$ another price vector $\hat \prices$ such that $\hat prices$ obtain at least a $(1-7\delta^{1/8}-\rho)$ fraction of the optimal for revenue for $\{v_i\}_{i\in[n]}$.
    \end{enumerate}
    If $u_{\text{min}}$ and $u_{\text{max}}$ are explicitly known, we can compute the distribution of $\{\hat v_i\}_{i\in[n]}$s and their support in time polynomial in the description of $\{v_i\}_{i\in[n]}$, the bit complexity of $u_{\text{min}}$ and $u_{\text{max}}$, and in $1/\delta$.
 \end{lemma}

 We may combine Lemmas~\ref{lemma:pricing-discretization}, \ref{lem:truncation}, and \ref{lem:discretization} with our framework to obtain a PTAS for the unbounded regular case of unit-demand pricing.
 In more detail, given distributions $\{v_i\}_{i\in[n]}$, we may apply Lemma~\ref{lem:truncation} to truncate the value distributions to new distributions $\{\hat v_i\}_{i\in[n]}$ to range $[\tfrac{\epsilon\alpha}{4n^4},\tfrac{4n^4\alpha}{\epsilon^3}]$.
 Given the truncated distributions $\{\hat v_i\}_{i\in[n]}$, we may discretize the value distributions using Lemma~\ref{lem:discretization} to obtain new distributions $\{\tilde v_i\}_{i\in[n]}$ with support polynomial in $n$ and $1/\epsilon$.
 Then given $\{\tilde v_i\}_{i\in[n]}$, we may apply Lemma~\ref{lemma:pricing-discretization} to obtain a restricted set of $\text{poly}(n,1/\epsilon)$ prices for the problem.
 We may then apply our framework to obtain near-optimal prices for $\{\tilde v_i\}_{i\in[n]}$ from the restricted set in time polynomial in $n$ and the number of prices and the number of mass points in $\{\tilde v_i\}_{i\in[n]}$ (which, in turn, are polynomial in $n$ for any fixed approximation error $\epsilon$).
 Then, given the prices computed by our framework, we may back out prices for the original instance using Lemmas~\ref{lem:discretization} and \ref{lem:truncation}.
 We conclude with the following:
 \begin{corollary}
     For any $\epsilon$, there is an algorithm which, for any independent regular distributions, finds prices within a $(1-\epsilon)$-factor of the optimal in time polynomial in the number of items $n$.
 \end{corollary}

%% file: appendix/appendix.tex
\section{Assortment Optimization is NP-hard}
\input{appendix/assortment_hardness}\label{sec:assortment-hard}

%% file: appendix/assortment_hardness.tex
We follow the same strategy as \citet{KPT24} and reduce from the \textsc{Integer Partition} problem.
\textsc{Integer Partition} takes a set of integers $\{c_1, \ldots, c_n\}$ and seeks a subset $S\subseteq [n]$ such that $\sum_{i\in S}c_i = \tfrac{1}{2}\sum_{i = 1}^n c_i$. 
We denote the sum of all integers by $C = \sum_i^n c_i$. 

We follow closely the reduction of \citet{KPT24}, and rather than formally replicating their proof, we instead provide a reduction to assortment optimization which has identical utilities for a given solution to the \citet{KPT24} reduction, and then refer the reader to their analysis of the optimal solution of the delegation reduction. 

We construct an instance of assortment optimization with $n+1$ items and an outside option as follows:

\begin{itemize}
    \item For item 0 (the outside option), we set $p_0 = 0$ and $v_0 = 0$
    \item For items $1, \ldots, n$, we set $p_i = 1$ and $v_i = \begin{cases}
        2 + \delta, &\text{w.p.} \frac{c_i}{M} - \frac{c_i^2}{2M^4(1-\frac{C}{2M})}\\
        2 - \delta, &\text{w.p.} \frac{c_i}{M}\\
        0, &\text{otherwise}
    \end{cases}$
    \item For items $n+1$, we set $p_{n+1} = M^2(1-\frac{C}{2M})+1+\delta$ and $v_{n+1} = \begin{cases}
        M^2(1-\frac{C}{2M})+2+\delta, &\text{w.p.} \frac{1}{2}\\
        0, &\text{w.p.} \frac{1}{2}
    \end{cases}$
\end{itemize}
Here, $M$ is a large constant picked to satisfy the constraints imposed in the proof of Theorem 2 of \citet{KPT24} ($M \geq 128n^3c_{max}^3$ suffices, where $c_{max} = \max_i c_i$).
Lastly, $\delta$ is a tiny constant used for tie-breaking and thus excluded from the principal utility computation.

We argue that this assortment optimization instance is structurally equivalent to the instance used by \citet{KPT24} to prove that the delegated choice problem is NP-hard:
\begin{itemize}
    \item For actions $1, \ldots, n$, we set $b_i = M^2(1-\frac{C}{2M})$ and $v_i = \begin{cases}
        1 + 2\delta, &\text{w.p.} \frac{c_i}{M} - \frac{c_i^2}{2M^4(1-\frac{C}{2M})}\\
        1, &\text{w.p.} \frac{c_i}{M}\\
        0, &\text{otherwise}
    \end{cases}$
    \item For action $n+1$, we set $b_{n+1} = 0$ and $v_{n+1} = \begin{cases}
        M^2(1-\frac{C}{2M})+1+\delta, &\text{w.p.} \frac{1}{2}\\
        0, &\text{w.p.} \frac{1}{2}
    \end{cases}$
\end{itemize}

Consider any subset of actions $S \subseteq \set{1, \ldots, n}$ in the delegation instance.
We claim that the distribution of the principal's utility in the delegation instance induced by $S$ is identical to that in the assortment optimization instance induced by item set $S \cup \set{0}$.
To see this, first observe that each action $i \in \set{1, \ldots, n+1}$ has the same probability distribution as the corresponding item in the assortment optimization instance.
Second, observe that for all realizations, the principal's utility is pairwise identical for both instances: if one of the first $n$ actions/items realizes to a positive value, the principal gets utility 1 in both instances; if the $(n+1)$-th item/action realizes to the positive value, the principal gets utility $M^2(1-\frac{C}{2M})+1+\delta$ in both instances; if all actions realize to 0, then the delegation instance will provide the principal a utility of 0, no matter what action is selected by the agent, while in the assortment optimization instance, the agent will select the outside option which provides the principal a utility of 0 as well.

Then, the NP-hardness proof for delegation (Theorem 2 of \citet{KPT24}) can be directly translated to the assortment optimization problem, completing the reduction.